\documentclass[5p, authoryear]{elsarticle}

\journal{Control Engineering Practice}
\makeatletter
\def\ps@pprintTitle{%
   \let\@oddhead\@empty
   \let\@evenhead\@empty
   \let\@oddfoot\@empty
   \let\@evenfoot\@oddfoot
}
\makeatother

\usepackage{natbib}
\usepackage{geometry}
\usepackage{graphicx}
\usepackage{dsfont}

\usepackage{amsmath, amssymb}
\usepackage{epstopdf}
\usepackage{float}
\usepackage{caption}
\usepackage{subcaption} 
\usepackage{booktabs}
\usepackage{enumitem}
\usepackage{amsthm} 
    \newtheorem{lemma}{Lemma}
    \newtheorem{assumption}{Assumption}

    \newtheorem{remark}{Remark}

    \newtheorem{proposition}{Proposition}

\usepackage[ruled, lined, longend, linesnumbered]{algorithm2e}
\SetKwInOut{Require}{Require}

\usepackage{multirow}
\usepackage{multicol}
\usepackage{arydshln}
\usepackage{xcolor}
\usepackage{hyperref}
\hypersetup{%
    colorlinks=true,
    linkcolor={blue!50!black},
    citecolor={blue!50!black},
    urlcolor={blue!50!black}
}

\usepackage{fancyhdr}
\def\R{\mathbb{R}} 

\newcommand {\T}{^{\top}} 
\def\vv#1{\mathrm{\mathbf{#1}}} 
\def\LBu{\underline{u}}
\def\UBu{\overline{u}}
\def\LBx{\underline{x}}
\def\UBx{\overline{x}}

\def\I{\mathbb{I}}
\def\({\left(}
\def\){\right)}
\def\diag{\texttt{diag}}
\def\blkdiag{\texttt{blkdiag}}

\def\tx{\tilde x}
\def\tu{\tilde u}
\def\th{\tilde h}

\begin{document}
\begin{frontmatter}

\title{\LARGE Artificial-reference tracking MPC with probabilistically validated performance on industrial embedded systems}

\author[a]{Victor~Gracia\corref{cor1}}%
\ead{vgracia@us.es}

\author[b]{Pablo~Krupa} 
\ead{pablo.krupa@imtlucca.it} 

\author[a]{Filiberto~Fele} 
\ead{ffele@us.es}

\author[a]{Teodoro~Alamo} 
\ead{talamo@us.es}

\cortext[cor1]{Corresponding author}

\affiliation[a]{organization={Departamento de Ingeniería de Sistemas y Automática, Universidad de Sevilla},
city={Seville},
country={Spain}}

\affiliation[b]{organization={IMT School for Advanced Studies},
city={Lucca},
country={Italy}}

\begin{abstract}
    Industrial embedded systems are typically used to execute simple control algorithms due to their low computational resources.
    Despite these limitations, the implementation of advanced control techniques such as Model Predictive Control (MPC) has been explored by the control community in recent years, typically considering simple linear formulations or explicit ones to facilitate the online computation of the control input.
    These simplifications often lack features and properties that are desirable in real-world environments.
    This article presents an efficient implementation for embedded systems of \emph{MPC for tracking with artificial reference}, solved via a recently developed structure-exploiting ADMM-based algorithm. This formulation is tailored to a wide range of applications by incorporating essential practical features at a small computational cost, including integration with an offset-free scheme, back-off parameters that enable constraint tightening, and soft constraints that preserve feasibility under disturbances or plant-model mismatch. This is accompanied with a framework for probabilistic performance validation of the closed-loop system over long-term operation. The applicability of the approach is illustrated on a Programmable Logic Controller (PLC), incorporated in a hardware-in-the-loop setup to control a nonlinear continuous stirred-tank reactor. The behavior of the closed-loop system is probabilistically validated with respect to constraint violations and the number of iterations required at each time step by the MPC optimization algorithm.
\end{abstract}


\end{frontmatter}

\thispagestyle{fancy}

\section{Introduction} \label{sec:introduction}

In the process industry, control devices typically prioritize reliability and robustness over computational and memory resources. An example of a widely used industrial embedded system is Programmable Logic Controllers (PLC), which, due to their hardware limitations, have mostly been used to implement simple control laws, such as PID controllers or ladder logic~\citep{alphonsus2016review}. 

Compared to these control strategies, advanced control techniques offer advantages that can be interesting in industrial scenarios. In particular, Model Predictive Control (MPC) has received much attention due to its ability to optimize plant operation while satisfying system constraints, such as actuator limits and safety requirements~\citep{Camacho_S_2013}. 
In MPC, the control input is obtained at each sample time by solving an optimization problem that, based on a prediction model of the system, determines a system trajectory satisfying some optimality criterion.
MPC has strong theoretical benefits, see~\cite{rawlings_model_2017}, and has shown very good results on real systems, e.g.,~\cite{ferramosca_insulin, Nubert_RAL_2020, Krupa_ECC_21, Moscato_CEP_2024}.
This performance comes at the expense of solving the MPC optimization problem at each time step.
Recent developments in efficient optimization methods, some specifically tailored to MPC problems~\citep{Frison_IFAC_2020, Lowenstein_LCSS_24, Spcies}, have resulted in solution times in the order of milliseconds when implemented on computers or microcontrollers.

Significant effort has been made towards the implementation of MPC in industrial systems, e.g., \cite{Valencia_CEP_2011}, where a suboptimal explicit MPC controller for PLCs is proposed; \cite{Goulart_embedded}, where custom embedded hardware architectures are developed for MPC deployment; or \cite{Krupa_TCST_21}, where MPC is tailored to PLC implementations accounting for relevant practical details.
A common feature in the literature is to consider simple MPC formulations with linear prediction model, quadratic objective function, and box constraints.
This choice is motivated by the fact that the resulting optimization problem of the MPC formulation is a simple instance of a Quadratic Programming (QP) problem. QP problems have been widely studied and can be efficiently solved, see~\cite{stellato_osqp_2020, ODonoghue_SCS_21}.
Among the different optimization methods used to solve QP problems, first-order methods~\citep{beck_first-order_2017} stand out as a particularly simple-to-implement class of methods that only require first-order information of the objective function and can leverage the problem sparsity, allowing memory-efficient implementations in embedded systems~\citep{FERREAU201713194}.

While these MPC formulations provide a good starting point for the implementation of MPC on embedded systems, a more suitable basis for the development of a formulation aimed at industrial applications is the family of~\emph{MPC for tracking} (MPCT) formulations with artificial reference~\citep{krupa_tutorial_2024, LIMON20082382}. Under nominal conditions, i.e., when the prediction model perfectly captures the real plant dynamics and no disturbances are present, MPCT formulations provide several benefits over simpler linear MPC controllers, such as an enlarged domain of attraction, or the ability to handle non-reachable references. Nevertheless, they also present some drawbacks that limit their viability in industrial environments. In particular, external disturbances or plant-model mismatch can lead to issues such as loss of feasibility, constraint violations, and steady-state offset when tracking constant references. Another feature often overlooked in the literature is the handling of joint state-input constraints, which can be desirable in many applications.

As a contribution, this article proposes the use of an MPCT formulation tailored to industrial embedded systems through an efficient implementation, augmented with features to address the aforementioned issues. The formulation considers joint state-input constraints that are enforced as soft constraints~\citep{Kerrigan2000SoftCA} to ensure recursive feasibility of the optimization problem. The soft constraints are formulated using an exact penalty approach specifically designed to preserve the computational efficiency of the solver. The MPC controller incorporates back-off parameters~\citep{constraint_tightening} for constraint tightening to improve constraint satisfaction of the closed-loop system.
For an efficient implementation, this work builds on the recently developed first-order MPC solver from~\cite{Gracia_MPCT_solver_24} and~\cite{Gracia_LCSS_24}, which was shown to provide low memory footprints and fast computation times. To allow piecewise-constant reference tracking without offset, the approach from \cite{MAEDER} is adopted, including a state-disturbance estimator that adjusts the MPC reference to compensate for plant-model mismatch and constant additive disturbances.

Although this formulation is suitable for a wide range of real applications, many of its features involve parameters that need careful adjustment to achieve adequate performance in real scenarios, where operating conditions are often not ideal. Consequently, parameter selection plays a critical role in ensuring the suitable and safe operation of the plant under real conditions. To address this challenge, this work proposes a framework that facilitates the configuration of the controller while providing probabilistic guarantees of performance. In particular, the probabilistic performance validation method in~\cite{karg_prob_validation} is extended to allow the validation of several aspects of the control setup. Furthermore, it is shown how the method can be applied for the validation of the long-term closed-loop system operation.

To illustrate the benefits of the proposed MPC scheme and the probabilistic validation approach, the controller is implemented on a Siemens S7-1500 PLC in a Hardware-in-the-Loop (HiL) setup to control a Continuous Stirred-Tank Reactor (CSTR), showing that advanced, practical MPC formulations can be deployed in industrial embedded systems with solution times in the order of a few seconds.
The solver is coded using standard IEC 61131-3 programming languages and adapted to PLC requirements, see~\citet[\S V]{Krupa_TCST_21}.
Our results allow to probabilistically certify that the closed-loop system does not violate the constraints beyond a given threshold, while also providing a probabilistic upper bound on the maximum number of iterations of the optimization solver. The latter ensures suitable computational times to obtain the control input, an essential factor in real-time control applications.

The remainder of this article is structured as follows. Section~\ref{sec:problem} presents the problem setting and the nominal MPCT formulation.
Section~\ref{sec:practical_formulation} details how the formulation is augmented for its real-world implementation, including the offset-free scheme, back-off parameters, and soft constraints.
Section~\ref{sec:prob_validation} illustrates the formal framework for the probabilistic validation of the control setup.
Section~\ref{sec:implementation} explains the structure-exploiting MPCT solver, suitable for embedded systems implementations.
Section~\ref{sec:HiL} presents the experimental results.
Section~\ref{sec:conclusions} closes the article.

\vspace{0.5em}
\noindent{\textbf{Notation:}
The identity matrix of size~$n$~is denoted as $\mathbf{I}_{n}$.
A $p$-dimensional column vector with all its elements equal to~$1$ is denoted as $\vv{1}_p$.
The $j$-th element of a vector $v\in\R^n$ is written as $v_{(j)}$.
The set of integers is denoted as~$\mathbb{I}$, with $\mathbb{I}_{i}^{j} = \{i, i+1, \dots, j\}$ for integers $i<j$. 
For a positive definite matrix $M$, $\| v\|_M \doteq \sqrt{v\T M v}$.
The function $\text{max}(\cdot)$ returns the maximum among the scalars received as input.
Given a vector $v \in\R^n$, $\diag(v)$ is a diagonal matrix with diagonal $v$. 
The block-diagonal matrix formed by placing the matrices $M_1, \dots, M_N$ along its diagonal is denoted by \blkdiag$(M_1,\dots,M_N)$. The column vector formed by the concatenation of column vectors $v_1,\dots,v_n$, possibly of different dimensions, is expressed as $(v_1,\dots,v_n)$. Symbols $\leq (<)$ and $\geq(>)$ denote component-wise inequalities when applied to vectors.
The indicator function on a set $\mathcal{C} \subseteq \R^{n}$ is denoted by $\delta_\mathcal{C} \colon \mathcal{C} \to \{0, +\infty\}$, i.e., $\delta_\mathcal{C}(x) = 0$ if $x \in \mathcal{C}$ and $\delta_\mathcal{C}(x) = +\infty$ if $x \not\in \mathcal{C}$.
Given a random vector $w\sim\mathcal{W}$, the notation $P_\mathcal{W}\{\cdot\}$ denotes the probability of an event involving $w$.

\section{Problem statement and MPC formulation} \label{sec:problem}

Consider the problem of controlling a possibly nonlinear time-invariant system, modeled by the discrete-time equations
\begin{subequations} \label{eq:non_lin_sys}
    \begin{align}
        x(k+1) &= f_p(x(k),u(k),d(k)), \label{eq:non_lin_x}\\
        y(k) &= g_p(x(k),d(k)), \label{eq:non_lin_y_meas}
    \end{align}
\end{subequations}
where $x(k)\in\R^{n}$, $u(k)\in\R^{m}$, $y(k) \in \R^{p}$ and $d(k)\in\R^{n_d}$ are the state, input, measured output and disturbance of the plant at sample time $k$, respectively. The state~$x(k)$ is not required to be measurable.
In the subsequent developments, $f_p$ and $g_p$ are only required to be continuous at all equilibrium points. Moreover, without loss of generality, it is assumed that $g_p(0,0)=0$, and that the origin is an equilibrium point, i.e., $f_p(0,0,0)=0$.

The control objective is to make~$y(k)$ track a piecewise constant reference~$y_r(k)\in\R^p$ while satisfying the system constraints
\begin{equation}\label{eq:non_lin_sys_constraints}
    (x(k),u(k)) \in \mathcal{Z}, \, \forall k\geq0,
\end{equation}
where $\mathcal{Z}\subseteq \R^{n+m}$ is a non-empty, closed and convex set that contains the origin in its interior.

This objective can be achieved by using a tracking MPC controller~\citep{rawlings_model_2017}.
As the controller is to be implemented on a limited hardware platform, solver efficiency is essential. While simple linear MPC formulations can be compatible with such requirement, their practical performance is limited by factors such as modeling errors or unknown disturbances. A contribution of this work is to provide an advanced, efficiently solvable formulation suitable for real control scenarios, achieved by incorporating key practical features.

As a starting point, consider the linear state-space prediction model 
\begin{subequations} \label{eq:lin_sys}
    \begin{align}
        x_{i+1} &= Ax_i+Bu_i, \label{eq:pred_model}\\
        y_i &= Cx_i,
    \end{align}
\end{subequations}
where $x_i\in\R^{n}$, $u_i\in\R^{m}$ and $y_i\in\R^{p}$ are the state, input, and output vectors of the model at prediction step~$i$, respectively.
This model captures the dynamics of~\eqref{eq:non_lin_sys} in a neighborhood of the origin, and is assumed to be controllable and observable. Matrices $A \in \R^{n \times n}$, $B \in \R^{n \times m}$ and $C \in \R^{p \times n}$ can be obtained using standard procedures, e.g., through system identification using experimental data, or by linearization of a high-fidelity nonlinear model of the real system.

\begin{assumption}\label{assumption}
System~\eqref{eq:lin_sys} is such that~$m \geq p$ and matrix~$\begin{bmatrix} A-\mathbf{I}_n & B \\ C & 0 \end{bmatrix}$ has full row rank.
\end{assumption}

The dynamics~\eqref{eq:lin_sys} are subject to the linear constraints%
\begin{subequations}\label{eq:lin_sys_constraints}
    \begin{align}
        &\LBx \leq x_i \leq \UBx,\label{eq:sys_constraints_x}\\ 
        &\LBu \leq u_i \leq \UBu,\label{eq:sys_constraints_u}\\
        &\underline{h} \leq E x_i + F u_i \leq \overline{h},\label{eq:sys_constraints_y}
    \end{align}
\end{subequations}
where $\LBx, \UBx \in \R^n$, $\LBu, \UBu \in \R^m$ and $\underline{h}, \overline{h} \in \R^{n_h}$ satisfy, respectively, $\LBx< \UBx$, $\LBu < \UBu$ and $\underline{h} < \overline{h}$, and $E\in\R^{n_h \times n}$ and $F\in\R^{n_h \times m}$ are the matrices of the coupled state-input constraints~\eqref{eq:sys_constraints_y}. Equations~\eqref{eq:lin_sys_constraints} characterize a non-empty convex set resulting from a linearization or approximation of the real system constraints~\eqref{eq:non_lin_sys_constraints}.

This work is based on the \emph{MPC for Tracking} (MPCT) with artificial reference formulation~\citep{LIMON20082382}. MPCT is a class of MPC formulations characterized by the addition of an artificial reference as part of the decision variables of its optimization problem. In the nominal case, i.e., when~\eqref{eq:lin_sys} matches~\eqref{eq:non_lin_sys} and there are no disturbances in~\eqref{eq:non_lin_sys}, this feature provides several benefits with respect to traditional MPC formulations, at the expense of some additional complexity in its underlying optimization problem.
These benefits include: increased domain of attraction, recursive feasibility even if the reference is changed online, and asymptotic stability to the admissible steady state that is ``closest'' to the given reference.

At sample time~$k$, the MPCT control law, for a given state-input reference pair~$(x_r(k),u_r(k)) \in \R^{n} \times \R^{m}$ and system state estimate~$\hat{x}(k) \in \R^n$, is derived from the solution of the optimization problem
\begin{subequations}\label{eq:MPCT}
	\begin{align}
		\min_{\substack{\vv{x,u}, \\x_s,u_s}} \; & V_o(x_s,u_s) + \sum_{i=0}^{N-1} \ell(x_i, u_i)  \label{eq:MPCT:cost} \\ 
		\textrm{s.t.} \; & x_{0} = \hat{x}(k),  \label{eq:MPCT:initial_constraint}\\
		& \eqref{eq:pred_model}, \ i \in \I_{0}^{N-2}, \label{eq:MPCT:model}\\
        & x_{s} = Ax_{N-1} + Bu_{N-1} \label{eq:MPCT:terminal},\\
        & x_{s} = Ax_{s} + Bu_{s}, \label{eq:MPCT:eq_point_constraint}\\
        & \eqref{eq:sys_constraints_x}, \ i \in \I_1^{N-1}, \label{eq:MPCT:ineq_x}\\
        & \eqref{eq:sys_constraints_u}, \ i \in \I_0^{N-1}, \label{eq:MPCT:ineq_u}\\ 
        & \eqref{eq:sys_constraints_y}, \ i \in \I_0^{N-1}, \label{eq:MPCT:ineq_y}\\
		& \LBx \leq x_{s} \leq \UBx, \ \LBu \leq u_{s} \leq \UBu, \label{eq:MPCT:ineq_xs_us} \\
        & \underline{h} \leq E x_s + F u_s \leq \overline{h}, \label{eq:MPCT:ineq_ys}
	\end{align}
\end{subequations}
where the stage cost function $\ell(\cdot)$ is defined as
\begin{equation*}
    \ell(x, u) \doteq \| x - x_s \|_Q^2 + \| u - u_s \|_R^2,
\end{equation*}
the offset cost function $V_o(\cdot)$ is given by
\begin{equation*}
    V_o(x_s, u_s) \doteq \|x_{s} - x_{r}(k)\|_{T}^2 +  \|u_{s} - u_{r}(k)\|_{S}^2,
\end{equation*}
and the weight matrices $Q \in \R^{n \times n}$, $R \in \R^{m \times m}$, $T \in \R^{n \times n}$ and $S \in \R^{m \times m}$ are positive definite.
The decision variables are the predicted states and inputs $\vv{x} = (x_0,\dots,x_{N-1})$, $\vv{u} = (u_0,\dots,u_{N-1})$, along the prediction horizon of length $N > 0$, together with the artificial reference $(x_s,u_s)\in\R^{n+m}$, forced to be an admissible steady state of~\eqref{eq:lin_sys} satisfying~\eqref{eq:lin_sys_constraints} by means of the constraints~\eqref{eq:MPCT:eq_point_constraint}, \eqref{eq:MPCT:ineq_xs_us} and \eqref{eq:MPCT:ineq_ys}\footnote{In the nominal case, the artificial reference constraints~\eqref{eq:MPCT:ineq_xs_us} and~\eqref{eq:MPCT:ineq_ys} require an arbitrarily small tightening for guaranteed stability~\citep{LIMON20082382}. This detail is omitted in formulation~\eqref{eq:MPCT}, as the inclusion of soft constraints in Section~\ref{sec:problem:soft} removes this requirement.}. Note that the terminal predicted state~$x_N$ is forced to be equal to~$x_s$ via the terminal constraint~\eqref{eq:MPCT:terminal}. 

The MPCT formulation~\eqref{eq:MPCT} allows to deal with references $(x_r, u_r)$ that are not admissible, i.e., that do not satisfy the model constraints~\eqref{eq:lin_sys_constraints}, and/or that are not steady states of~\eqref{eq:lin_sys}.
Under nominal conditions, if a constant reference $(x_r, u_r)$ is admissible, the closed-loop system will asymptotically converge to it; otherwise, it will asymptotically converge to the admissible steady state $(x^\circ, u^\circ)$ that minimizes $\| x^\circ - x_r \|_T^2 + \| u^\circ - u_r \|_S^2$, i.e., to the closest admissible steady state, as measured by the matrices~$T$ and~$S$. An in-depth review of the literature on MPCT can be found in the tutorial article~\cite{krupa_tutorial_2024}.

\begin{remark}
While the tracking objective is stated in terms of the output of~\eqref{eq:non_lin_sys}, the formulation~\eqref{eq:MPCT} requires the state-input reference pair $(x_r(k),u_r(k))$. This is clarified in Section~\ref{sec:offset_free}, where an offset-free strategy allowing to obtain $(x_r(k),u_r(k))$ from the output reference~$y_r(k)$ is included.
\end{remark}

\section{Practical MPCT formulation}\label{sec:practical_formulation}

Despite the aforementioned benefits of the formulation~\eqref{eq:MPCT} compared to simpler MPC formulations, its performance may still not be adequate in real-world applications, where the nominal conditions often do not hold. Model mismatch, discrepancies between the model constraints~\eqref{eq:lin_sys_constraints} and the actual system constraints~\eqref{eq:non_lin_sys_constraints}, and external disturbances can lead to steady-state offset when tracking constant references, or even to infeasibility of the MPCT optimization problem, in which case the controller is unable to provide a control input for the system. The following section shows how the MPCT formulation~\eqref{eq:MPCT} can be tailored to real-world control applications by including key practical features, compatibly with industrial embedded hardware capabilities.

To this end, the inclusion of three elements in the MPCT formulation~\eqref{eq:MPCT} are now discussed: an offset-free strategy to cancel the offset derived from plant-model mismatch and disturbances; back-off parameters that tighten the model constraints~\eqref{eq:lin_sys_constraints} to make the controlled system~\eqref{eq:non_lin_sys} satisfy the actual constraints~\eqref{eq:non_lin_sys_constraints}; and the softening of most of the MPC inequality constraints, resulting in a formulation whose optimization problem is always feasible.

\subsection{Disturbance rejection and offset-free MPC} \label{sec:offset_free}

To address disturbance rejection and offset cancelation, the approach described in \cite{MAEDER} is used. Model~\eqref{eq:lin_sys} is augmented as
\begin{subequations} \label{eq:augmented_system}
    \begin{align}
        x_{i+1} &= A x_i + B u_i + B_d d_i,\\
        d_{i+1} &= d_i,\\
        y_i &= C x_i+d_i,
    \end{align}
\end{subequations}
where $B_d \in \R^{n\times p}$, and the disturbance~$d_i\in\R^{p}$ has the same dimension as the measured output.
At each time step $k$, a Luenberger observer is used to update the estimates~$\hat{x}$ and~$\hat{d}$ of~$x$ and~$d$ as
\begin{equation} 
    \begin{aligned}
        \begin{bmatrix}\label{eq:observer}
           \hat{x}(k+1) \\ \hat{d}(k+1)
        \end{bmatrix}
        = &
        \begin{bmatrix}
            A & B_d \\
            0 & \mathbf{I}_p
        \end{bmatrix}
        \begin{bmatrix}
            \hat{x}(k) \\
            \hat{d}(k)
        \end{bmatrix} +
        \begin{bmatrix}
            B \\ 0
        \end{bmatrix} u(k) \\ &+
        \begin{bmatrix}
            L_x \\ L_d
        \end{bmatrix}
         \(C \hat{x}(k) + \hat{d}(k) - y(k)\),
    \end{aligned}
\end{equation}
where $L_x \in \R^{n \times p}$ and $L_d \in \R^{p \times p}$ are the observer gains, which must be designed so as to ensure the stability of the observer~\citep{MAEDER}.

At each sample time~$k$, the current state estimate~$\hat{x}(k)$ is provided to the MPCT controller, and the estimated disturbance~$\hat{d}(k)$ is used to determine the MPCT reference pair $(x_r(k),u_r(k))$ by solving the linear system
\begin{equation} \label{eq:ref_sys}
    \begin{bmatrix}
        A-\mathbf{I}_n & B \\
        C & 0
    \end{bmatrix}
    \begin{bmatrix}
        x_r(k) \\ u_r(k)
    \end{bmatrix} =
    \begin{bmatrix}
        -B_d \hat{d}(k) \\
        y_r(k) - \hat{d}(k)
    \end{bmatrix}.
\end{equation}
Assumption~\ref{assumption} implies that system~\eqref{eq:ref_sys} is solvable for any $y_r(k)$ and~$\hat{d}(k)$. However, the solution may not be unique or may not satisfy the system constraints~\eqref{eq:lin_sys_constraints}. Since the MPCT can handle nonadmissible references $(x_r(k),u_r(k))$, any solution of~\eqref{eq:ref_sys} can directly serve as a reference. When the solution of~\eqref{eq:ref_sys} is not unique, one can select a specific solution, e.g., the one that minimizes the norm of~$u_r(k)$.

\begin{remark}
The results in~\cite{MAEDER} consider a more general augmented model~\eqref{eq:augmented_system} that includes the case in which only a subset of the measured outputs are required to be tracked without  offset. However, for clarity of presentation, this work focuses on the case where all measured outputs are tracked. In any case, the dimension of~${d}$ in~\eqref{eq:augmented_system} can always be selected as~$p$, simplifying the design of the observer~\eqref{eq:observer}.
\end{remark}

\subsection{Tightening of constraints} \label{sec:back-off}

The use of the formulation~\eqref{eq:MPCT} to control the system~\eqref{eq:non_lin_sys} does not guarantee the satisfaction of the constraints~\eqref{eq:non_lin_sys_constraints}. This is due to: the possible discrepancy between the constraints of the MPC model~\eqref{eq:lin_sys_constraints} and those of the real system~\eqref{eq:non_lin_sys_constraints}, the inability of the controller to predict the actual evolution of the plant given the mismatch between the real system~\eqref{eq:non_lin_sys} and the model~\eqref{eq:lin_sys}, the presence of external disturbances, and the inherent estimation error of the observer~\eqref{eq:observer}.
To address these issues, the model constraints~\eqref{eq:lin_sys_constraints} are tightened by means of back-off parameters~\citep{constraint_tightening}, in order to ensure that the plant~\eqref{eq:non_lin_sys} actually satisfies its true constraints~\eqref{eq:non_lin_sys_constraints} during operation.
In particular, the state and coupled state-input constraints of~\eqref{eq:MPCT} can be tightened by introducing nonnegative back-off parameters $\underline{\eta}_{x}, \overline{\eta}_{x} \in \R^{n}$, $\underline{\eta}_{h}, \overline{\eta}_{h} \in \R^{n_h}$. The resulting set of tightened model constraints is
\begin{subequations}\label{eq:tightened_constraints}
    \begin{align}
         &\underline{x} + \underline{\eta}_{x} \leq x_i \leq \overline{x}-\overline{\eta}_{x},\, i\in\mathbb{I}_1^{N-1}, \\
         &\underline{x} + \underline{\eta}_{x} \leq x_s \leq \overline{x}-\overline{\eta}_{x}, \\
         &\underline{h} + \underline{\eta}_{h} \leq E x_i + F u_i \leq \overline{h} - \overline{\eta}_{h},\, i\in\mathbb{I}_0^{N-1},\\
         &\underline{h} + \underline{\eta}_{h} \leq E x_s + F u_s \leq \overline{h} - \overline{\eta}_{h},
    \end{align}
\end{subequations}
where $\underline{x} + \underline{\eta}_{x} < \overline{x}-\overline{\eta}_{x}$ and $\underline{h} + \underline{\eta}_{h} < \overline{h} - \overline{\eta}_{h}$.
Although possible, input constraints tightening is not considered, as they typically represent real actuator limits.

The selection of the back-off parameters $\underline{\eta}_{x}, \overline{\eta}_{x} \in \R^{n}$, $\underline{\eta}_{h}, \overline{\eta}_{h} \in \R^{n_h}$ usually involves a trade-off. Small values might be insufficient to avoid excessive violations of the actual constraints. Conversely, large values might lead to conservative behavior. This might also cause the reference pair~$(x_r(k),u_r(k))$ provided by~\eqref{eq:ref_sys} to lie outside the tightened MPC constraints, resulting in steady-state offset of the closed-loop system. In Section~\ref{sec:prob_validation}, a formal validation approach enabling the selection of the values of the back-off parameters in~\eqref{eq:tightened_constraints} is presented, providing probabilistic bounds on the satisfaction of the constraints~\eqref{eq:non_lin_sys_constraints} during online operation.

\subsection{Soft constraints} \label{sec:problem:soft}

One of the well-known limitations of MPC when considering hard constraints, such as in~\eqref{eq:MPCT}, is the possible infeasibility of the optimization problem for a given~$\hat{x}(k)$.
In this situation, the MPC controller cannot provide a control input, requiring the use of an auxiliary controller.
This situation is expected when using linear MPC to control a nonlinear system, either due to disturbances, plant-model mismatch or estimation errors.
To address this issue, hard inequality constraints are relaxed into \emph{soft constraints}~\citep{Kerrigan2000SoftCA, Wabersich2022980}, whose violation is penalized in the cost function. This ensures that the optimization problem is always feasible, enabling the controller to always compute a control input.

In this article, soft constraints are encoded as proposed in~\cite{Gracia_LCSS_24}. As an example, the constraint~\eqref{eq:sys_constraints_x} can be softened by replacing it with a nonsmooth term $\gamma_\beta(x_i ; \overline{x}, \underline{x})\doteq \sum_{j=1}^{n} \beta_{(j)} \text{max}({x_i}_{(j)}-\overline{x}_{(j)},\underline{x}_{(j)}-{x_i}_{(j)},0)$
in the objective function, where $\beta \in \R^{n}$ is a positive vector of coefficients penalizing the constraint violation of each element of~$x_i$. The remaining inequality constraints of the formulation~\eqref{eq:MPCT} can be relaxed in the same manner. This encoding retains the simple structure of the MPCT optimization problem, which is exploited by the numerical solver, as detailed in Section~\ref{sec:implementation}.

\subsection{Practical MPCT formulation} \label{sec:final_formulation}
This subsection summarizes the formulation to be implemented on the embedded system, which includes the features presented in the previous subsections.
Defining $h_i \doteq E x_i + F u_i, \, i\in\mathbb{I}_0^{N-1}$, $h_s \doteq E x_s + F u_s$, and the vector
$\vartheta \doteq (h_0,x_1,u_1,h_1, {\dots}, x_{N-1},u_{N-1},h_{N-1}, x_s,u_s,h_s)\in\R^{n_\vartheta},$ where $n_\vartheta = n_h+N(n+m+n_h)$, the MPCT formulation~\eqref{eq:MPCT} can be cast in its practical variant as
\begin{subequations}\label{eq:MPCT_soft}
	\begin{align}
		\min_{\substack{\vv{x,u}, \\x_s,u_s}} \; &
	    V_o(x_s,u_s) + \gamma_\beta(\vartheta ; \overline{\vartheta}, \underline{\vartheta}) + \sum_{i=0}^{N-1} \ell(x_i, u_i) \label{eq:MPCT_soft:cost} \\
		\textrm{s.t.} \; & x_0 = \hat{x}(k), \label{eq:MPCT_soft:initial_constraint}\\
		& x_{i+1} = Ax_{i} + Bu_{i} + B_d \hat{d}(k), \ i \in \mathbb{I}_{0}^{N-2}, \label{eq:MPCT_soft:model}\\
		& x_{s} = Ax_{N-1} + Bu_{N-1} + B_d \hat{d}(k), \label{eq:MPCT_soft:terminal}\\
		& x_{s} = Ax_{s} + Bu_{s} + B_d \hat{d}(k), \label{eq:MPCT_soft:eq_point_constraint}\\
        & \LBu \leq u_0 \leq \UBu, \label{eq:MPCT_soft:hard_u0}
	\end{align}
\end{subequations}
where it is worth noting that~\eqref{eq:MPCT_soft:initial_constraint} uses the estimated state~$\hat{x}(k)$ provided by the observer~\eqref{eq:observer}. The term used to include soft constraints $\gamma_\beta(\cdot)$ penalizes the violation of constraints for the variables included in~$\vartheta$. 
Note that $\vartheta$ collects all the decision variables except for~$x_0$, as it does not present box constraints, and~$u_0$, which keeps its hard constraint~\eqref{eq:MPCT_soft:hard_u0} to provide a control action satisfying the physical limits of the actuators. Additionally, $\vartheta$ includes the variables $h_i, \, i \in \mathbb{I}_0^{N-1}$, and $h_s$, which are used to soften the coupled state-input constraints defined in~\eqref{eq:MPCT:ineq_y} and~\eqref{eq:MPCT:ineq_ys}.
The vectors $\underline{\vartheta},\overline{\vartheta}\in\R^{n_\vartheta}$ collect the bounds of the box constraints in~$\vartheta$, considering also the tightening introduced by the back-off parameters, i.e.,
\begin{equation*}
    \begin{aligned}
        &\underline{\vartheta} \doteq (\underline{h} + \underline{\eta}_{h}, \, \underline{x} + \underline{\eta}_{x}, \,\underline{u}, \, \underline{h} + \underline{\eta}_{h}, \, \dots, \, \underline{x} + \underline{\eta}_{x}, \,\underline{u}, \, \underline{h} + \underline{\eta}_{h}),\\
        &\overline{\vartheta} \doteq (\overline{h} - \overline{\eta}_{h}, \, \overline{x} - \overline{\eta}_{x}, \,\overline{u}, \, \overline{h} - \overline{\eta}_{h}, \, \dots, \, \overline{x} - \overline{\eta}_{x}, \,\overline{u}, \, \overline{h} - \overline{\eta}_{h}).
    \end{aligned}
\end{equation*}

Lemma~\ref{lemma} states a sufficient condition on~$N$ under which the MPC optimization problem~\eqref{eq:MPCT_soft} is feasible. Additionally, there exists a sufficiently large $\beta$ such that the optimal solution of the soft-constrained formulation~\eqref{eq:MPCT_soft} coincides with the optimal solution of its equivalent hard-constrained version whenever the latter is feasible; this is due to the exact penalty property of $\gamma_\beta(\cdot)$~\citep[Thm. 14.3.1]{Fletcher_PMO_2000}, \citep{Kerrigan2000SoftCA}.
\begin{lemma}\label{lemma}
Problem~\eqref{eq:MPCT_soft} is feasible $\forall (\hat{x}(k),\hat{d}(k)) \in \R^n \times \R^p$ if~$N>n_c$, where~$n_c$ is the controllability index of~\eqref{eq:lin_sys}.
\end{lemma}
\begin{proof}
From Assumption~\ref{assumption},~$\begin{bmatrix} A-\mathbf{I}_n & B \end{bmatrix}$ has full row rank. Therefore,~$\forall \hat{d}(k) \in \R^p$, $\exists (x_s,u_s)$ satisfying~\eqref{eq:MPCT_soft:eq_point_constraint}. Define $e_i \doteq x_i-x_s$ and~$v_k \doteq u_i - u_s$. Then, \eqref{eq:MPCT_soft:initial_constraint} can be rewritten as~$e_0=\hat{x}(k)-x_s$,~\eqref{eq:MPCT_soft:model} as~$e_{i+1} = A e_i + B v_i, \; i \in \mathbb{I}_0^{N-2}$,~\eqref{eq:MPCT_soft:terminal} as~$0 = A e_{N-1} + B v_{N-1}$, and~\eqref{eq:MPCT_soft:hard_u0} as~$\LBu \leq v_0 + u_s \leq \UBu$. Take any~$v_0$ such that~$\LBu \leq v_0 + u_s \leq \UBu$, which determines~$e_1 = A e_0 + B v_0$. Since the pair~$(A,B)$ is controllable, starting from $e_1$, system~$e_{i+1} = A e_i + B v_i$ can reach any~$e\in\R^n$ in (at most) $n_c$ steps~\citep[\S~6.6]{Chen1995LinearST}, implying that~$\exists \{v_i\}_{i=1}^{n_c^*}$, with $n_c^* \leq n_c$, such that~$e_{n_c^*+1}=0$. Since~$N > n_c$, $\exists \{v_i\}_{i=1}^{N-1}$ such that~$0 = e_N = A e_{N-1} + B v_{N-1}$. Existence of sequences~$\{x_i\}_{i=0}^{N-1}$ and~$\{u_i\}_{i=0}^{N-1}$ satisfying~\eqref{eq:MPCT_soft:model} and~\eqref{eq:MPCT_soft:terminal} follow from undoing the change of variable.
\end{proof}

\begin{remark}
Replacing hard constraints in~\eqref{eq:MPCT_soft} with the soft constraint term~$\gamma_\beta(\vartheta ; \overline{\vartheta}, \underline{\vartheta})$ may lead to constraint violations in cases where the hard-constrained problem is feasible and has active constraints in its optimal solution. Although this effect can be suppressed by taking a sufficiently large value of the soft constraint parameter~$\beta$, this may come at the cost of an increase in the number of iterations of the optimization method in cases where the hard-constrained problem is not feasible. A balance between solve time of~\eqref{eq:MPCT_soft} and possible suboptimality in terms of unnecessary constraint violations requires careful tuning of~$\beta$. As shown in the numerical results presented in Section~\ref{sec:HiL}, the probabilisitic validation method described in the following section can be used to guide this tuning process, allowing the selection of a value of~$\beta$ that provides sufficiently small constraint violations and computation times.
\end{remark}

\section{Probabilistic validation of the controller} \label{sec:prob_validation}

In the previous section, an advanced MPC formulation with features that enable its applicability in real control scenarios was described. However, selecting the controller parameters so that the real plant operates as desired is not trivial. A practical approach is to define several controllers based on different parameter values, test their behavior on a number of experiments, and select the most satisfactory one. Still, choosing a configuration is challenging in many cases, as the selected controller may be unsuitable in scenarios other than those considered in the experiments, or it may not provide a good balance in every aspect of interest, e.g., control response, constraint satisfaction or computation time.

Furthermore, parameter selection is crucial for guaranteeing stability in the presence of uncertainty. While robust techniques such as tube-based MPC provide theoretical robust stability guarantees~\citep{LANGSON,ALVARADO_ROBUST}, they often rely on worst-case uncertainty bounds that can lead to conservative performance, or require complex offline design of invariant sets. In contrast, the back-off parameters in~\eqref{eq:tightened_constraints} are used to modify the model constraints~\eqref{eq:lin_sys_constraints}, with the objective that the controller makes the real system~\eqref{eq:non_lin_sys} satisfy its true constraints~\eqref{eq:non_lin_sys_constraints}. As it will be shown, these parameters can be selected based on empirical data without explicitly determining bounds on the discrepancy between~\eqref{eq:non_lin_sys} and~\eqref{eq:lin_sys} or on the disturbance~$d$ in~\eqref{eq:non_lin_sys}, as typically required in tube-based MPC. This selection allows us to certify with high confidence that the controller maintains stability under realistic conditions while avoiding excessive conservatism.

To enable this selection, a performance validation approach is proposed to select a controller that satisfies, in a probabilistic sense, multiple performance requirements \citep{mammarella}.
In particular, the main result in~\cite{karg_prob_validation} is extended to enable the selection of one among $M \geq 1$ candidate controllers with probabilistic guarantees across $K\geq1$ performance metrics. 

Let $w\in\R^{n_w}$ be a vector collecting the uncertain variables that uniquely determine\footnote{``Uniquely determine'' refers to the fact that a given realization~$w$ always results in the same system trajectory.} a trajectory of the system when in closed-loop with a controller $\kappa$ (e.g., system disturbances, communication delays, etc.). Therefore, the size of $w$ is typically proportional to the duration of the trajectory, considered here of $N_{t}$ time steps. It is assumed that~$w$ is a random vector following some (unknown) probability distribution~$\mathcal{W}$. Let $\phi(w;\kappa,N_t):\R^{n_w} \rightarrow \R$  be a closed-loop performance indicator \citep[Definition~1]{karg_prob_validation} that returns a scalar evaluating the performance of the controller~$\kappa$ under the realization~$w\sim\mathcal{W}$. In this work,~$\phi(\cdot)$ is defined such that lower values indicate better performance. In Section~\ref{sec:HiL}, performance indicators are defined to assess the satisfaction of the constraints~\eqref{eq:non_lin_sys_constraints} and the maximum number of iterations of the optimization solver, for which probabilistic bounds will be established.

Given $N_s$ independent and identically distributed (i.i.d.) samples $w_j\sim\mathcal{W}$, $j \in \I_1^{N_s}$, $M$ controllers $\kappa_i$, $i\in\mathbb{I}_1^M$, and $K$ performance indicators $\phi^{\ell}(\cdot)$, $\ell \in \mathbb{I}_1^K$, let~$\phi_i^{\ell}(w_j) \doteq \phi^{\ell}(w_j;\kappa_i,N_t)$ denote the performance of the controller $\kappa_i$ under the uncertainty realization $w_j$, as evaluated by the indicator $\phi^{\ell}$. In what follows, the notation~$\phi_{i}^{\ell,[r]}$ refers to the $r$-th worst performance among the $N_s$ samples, i.e., the $r$-th largest value of the set $\{\phi_i^{\ell}(w_1),\dots,\phi_i^{\ell}(w_{N_s})\}$.

The formal validation procedure, applicable to any set of controller parameters, is now presented.
\begin{proposition} \label{proposition}
Given the controllers $\kappa_i$, $i\in \I_1^M$, and the performance indicators $\phi^{\ell}(\cdot)$, $\ell\in\I_1^K$, suppose that $N_s$ \text{i.i.d.} scenarios $w_j \sim \mathcal{W}$, $j\in \I_1^{N_s}$, are generated.
Then, for any choice of the integer $1 \leq r \leq N_s$, it holds with probability no smaller than $1-\delta$ that
\begin{equation} \label{eq:probability}
    \text{P}_\mathcal{W}\{\phi_i^{\ell}(w)>\phi_{i}^{\ell,[r]}\} \leq \epsilon, \ i\in\I_1^M, \ \ell\in\I_1^K,    
\end{equation}
provided that
\begin{equation} \label{eq:pre_samples_condition}
 \sum_{q=0}^{r-1} \binom{N_s}{q} \epsilon^q(1-\epsilon)^{N_s-q} \leq \frac{\delta}{MK}.
\end{equation}
\end{proposition}

\begin{proof}
Proposition~\ref{proposition} follows directly from Theorem~1 in \cite{karg_prob_validation}. It suffices to observe that considering $K\geq2$ performance indicators implies accounting for a greater number of statistical events with respect to the original theorem.
\end{proof}

It can be verified that~\eqref{eq:pre_samples_condition} is satisfied if
\begin{equation} \label{eq:samples_condition}
 N_s\geq \frac{1}{\epsilon}\(r-1+\mathrm{ln}\frac{MK}{\delta}+\sqrt{2(r-1)\mathrm{ln}\frac{MK}{\delta}}\).
\end{equation}
Condition~\eqref{eq:samples_condition} can be used to determine a sufficient number~$N_s$ of closed-loop experiments such that the double probability level defined in Proposition~\ref{proposition} holds. The outer level is the confidence that the probability bound~\eqref{eq:probability} is valid. This is lower bounded by $1-\delta$ and accounts for the risk that the $N_s$ samples might not be representative of~$\mathcal{W}$. The inner level~\eqref{eq:probability} is the probability of obtaining a performance (as measured by~$\phi^{\ell}$) worse or equal to the $r$-th worst among the $N_s$ validation experiments (i.e., $\phi_{i}^{\ell,[r]}$), which is at most $\epsilon$. In other words, $N_s$ selected as in~\eqref{eq:samples_condition} ensures that for given controller $\kappa_i$, $i\in \mathbb{I}_1^M$, and performance indicator $\phi^{\ell}$, $\ell \in \mathbb{I}_1^K$, when taking a new realization $w\sim\mathcal{W}$, 
the probability of obtaining a performance better than the empirical threshold~$\phi_{i}^{\ell,[r]}$
holds with probability larger than $1-\epsilon$, with confidence at least~$1-\delta$. It is also worth noting that~$\epsilon$ and~$\delta$ are independent design parameters, although they both affect the minimum $N_s$ satisfying~\eqref{eq:samples_condition}. In practice, since~$M$, $K$ and~$\delta$ have small effect on the minimum~$N_s$ satisfying~\eqref{eq:samples_condition}, the value of~$\delta$ can be set to provide a very high confidence bound. Finally, note that the parameter~$r$ is selected so as to adjust the conservativeness of the probabilistic bounds, at the expense of requiring a larger number of experiments (as visible from~\eqref{eq:samples_condition}).

\begin{remark}
    When one of the $M$ controllers is selected after applying Proposition~\ref{proposition} for $K\geq2$, the probability that a new $w\sim\mathcal{W}$ results in a better performance than the $r$-th worst case out of the $N_s$ experiments is guaranteed (with confidence at least $1-\delta$) to be at least $1-\epsilon$ for each performance indicator individually. However, no guarantee is provided that, for a given $w\sim\mathcal{W}$, all $K$ performance indicators will simultaneously be better than their respective $r$-th worst cases. Nevertheless, when $\epsilon$ is small, simultaneous satisfaction is often observed empirically.
\end{remark}

\begin{remark}
Even though there is no closed-form expression for the domain of attraction of the proposed control scheme, note that the random variable~$w$ characterizing a closed-loop experiment includes the initial conditions $(x(0),\hat{x}(0),\hat{d}(0))$. A probabilistic validation of constraint satisfaction and convergence of the closed-loop system is obtained if performance indicators that assess these properties are selected. This, in turn, provides a probabilistic guarantee that the initial conditions~$(x(0),\hat{x}(0),\hat{d}(0))$ of any new closed-loop experiment belongs to the domain of attraction of the controller.
\end{remark}

\section{Implementation of the MPC solver} \label{sec:implementation}

This section presents the practical algorithmic implementation of the MPCT formulation~\eqref{eq:MPCT_soft}. To this end, this work combines the solver of~\cite{Gracia_MPCT_solver_24} based on the Alternating Direction Method of Multipliers (ADMM) algorithm \citep{Boyd_ADMM} tailored to exploit the MPCT problem structure, with the approach presented in~\cite{Gracia_LCSS_24} to implement the soft constraints. By exploiting the MPCT problem structure, significant reductions in memory requirement and computation time per iteration are achieved compared to general-purpose QP solvers such as \cite{stellato_osqp_2020, ODonoghue_SCS_21}, even when these leverage sparse matrix representations and sparse linear algebra subroutines.

Since the numerical results in Section~\ref{sec:HiL} are of primary interest, a brief review is provided on how the ADMM algorithm is used to solve~\eqref{eq:MPCT_soft}, and how the soft constraints are handled without affecting the problem structure, thus preserving efficiency.
For a detailed explanation of the solver, see~\cite{Gracia_MPCT_solver_24, Gracia_LCSS_24}.
An implementation of this solver in C programming language can be found in the open-source Matlab package SPCIES~\citep{Spcies}.

Consider the optimization problem
\begin{subequations} \label{eq:ADMM}
    \begin{align}
        \min_{z,v} & \ f(z)+g(v)\\
        \rm{s.t.} & \ Dz=v, \label{eq:ADMM:eq}
    \end{align}
\end{subequations}
where $z\in\R^{n_z}$ and $v\in\R^{n_v}$ are the optimization variables, $D\in\R^{n_v\times n_z}$, and $f:\R^{n_z}\rightarrow\R\cup\{+\infty\}$, $g:\R^{n_v}\rightarrow\R\cup\{+\infty\}$ are convex, closed, proper functions.
For problem~\eqref{eq:ADMM}, ADMM considers an augmented Lagrangian $\mathcal{L}_\rho:\R^{n_z} \times \R^{n_v} \times \R^{n_v} \rightarrow \R\cup\{+\infty\}$ given by
\begin{equation}\label{eq:aug_lagrangian}
    \mathcal{L}_\rho(z,v,\lambda) = f(z)+g(v)+\lambda\T (Dz-v)+ \frac{\rho}{2}\|Dz-v\|_2^2,
\end{equation}
where $\lambda\in\R^{n_v}$ is the vector of dual variables and the constant $\rho>0$ is a regularization parameter.
The steps of ADMM are given by Algorithm~\ref{alg:ADMM}, which, at each iteration~$l$, first minimizes~\eqref{eq:aug_lagrangian} with respect to $z$, then with respect to $v$, and finally updates the dual variables $\lambda$.
Starting from any $(z^0, v^0, \lambda^0) \in \R^{n_z} \times \R^{n_v} \times \R^{n_v}$, the iterates $(z^l, v^l, \lambda^l)$ of Algorithm~\ref{alg:ADMM} asymptotically converge to an optimal solution $(z^*,v^*,\lambda^*)$ of~\eqref{eq:ADMM}, assuming that one exists \citep{ADMM_convergence}.
In practice, however, the algorithm is terminated when the exit conditions in Step~\ref{alg:ADMM:exit_cond} are satisfied, where $\epsilon_p > 0$ and $\epsilon_d > 0$ are exit tolerances.

\begin{algorithm}[t]
	\DontPrintSemicolon
	\caption{ADMM} \label{alg:ADMM}
	\Require{$v^{0}$, $\lambda^{0}$, $\rho>0$, $\epsilon_{p}>0$, $\epsilon_{d}>0$}
	$l \gets 0$\;
	\Repeat{{$\|Dz^{l}-v^{l}\|_{\infty} \leq \epsilon_{p}$ and $\|v^{l}-v^{l-1}\|_{\infty} \leq \epsilon_{d}$} \label{alg:ADMM:exit_cond}}{
		$z^{l+1} \gets \displaystyle \arg \min_{z} \mathcal{L}_{\rho}(z, v^{l}, \lambda^{l})$\; \label{alg:ADMM:step_z}
		$v^{l+1} \gets \displaystyle \arg \min_{v} \mathcal{L}_{\rho}(z^{l+1}, v, \lambda^{l})$\; \label{alg:ADMM:step_v}
		$\lambda^{l+1} \gets \lambda^{l} + \rho(Dz^{l+1} - v^{l+1})$\;
		$l \gets l+1$\;
	}		
	\KwOut{$\tilde{z}^{*} \gets z^{l}$, $\tilde{v}^{*} \gets v^{l}$, $\tilde{\lambda}^{*} \gets \lambda^{l}$} 
\end{algorithm}

To pose the MPCT problem~\eqref{eq:MPCT_soft} in the form of~\eqref{eq:ADMM}, let
$$z \doteq (x_0, u_0, \dots, x_{N-1}, u_{N-1}, x_s, u_s) \in \R^{n_z}, $$
which collects the decision variables of~\eqref{eq:MPCT_soft}, with $n_z = (N+1)(n+m)$;
$$v \doteq (\tx_0, \tu_0, \th_0, \dots, \tx_{N-1}, \tu_{N-1}, \th_{N-1}, \tx_{s}, \tu_{s}, \th_{s}) \in \R^{n_v},$$
with $n_v = (N+1)(n+m+n_h)$, which contains copies of the variables in $z$ to decouple the equality and inequality constraints of~\eqref{eq:MPCT_soft}, and includes variables~$\th$ to handle the joint state-input constraints; 
$$D = \blkdiag\( \begin{bmatrix}
    \mathbf{I}_n & 0\\
    0 & \mathbf{I}_m\\
    E & F
\end{bmatrix},\dots,\begin{bmatrix}
    \mathbf{I}_{n} & 0\\
    0 & \mathbf{I}_{m}\\
    E & F
\end{bmatrix}\),$$ which relates vectors $z$ and $v$ through the constraint \eqref{eq:ADMM:eq};
$$f(z) = \frac{1}{2} z\T H z + q\T z + \delta_{\{z \colon G z = b\}}(z),$$
where, defining $m_z = (N+1)n$,
\begin{subequations} \label{eq:ADMM:MPCT:ingredients}
\begin{align} 
        H &= \begin{bmatrix}
				Q & 0 & \cdots & -Q & 0\\
				0 & R & \cdots & 0 & -R\\
				0 & 0 & \ddots & \vdots & \vdots\\
				-Q & 0 & \cdots & NQ+T & 0\\
				0 & -R & \cdots & 0 & NR+S
			\end{bmatrix} \in \R^{n_z \times n_z}, \label{eq:ADMM:MPCT:ingredients:H} \\
        q &= - (0, \dots, 0, Tx_r(k), Su_r(k))\in\R^{n_z},\\
 		G &=\begin{bmatrix}
 				\mathbf{I}_n & 0 & 0 & 0 & \cdots & 0\\
 				A & B & -\mathbf{I}_n & 0 & \cdots &  0\\
 				0 & \ddots & \ddots & \ddots & 0 & \vdots\\
 				0 & 0 & A & B & -\mathbf{I}_n & 0\\
 				0 & 0 & 0 & 0 & (A-\mathbf{I}_n) & B
 			\end{bmatrix} \in\R^{m_z \times n_z}, \label{eq:ADMM:MPCT:ingredients:G}\\
        b &= (\hat{x}(k) , - B_d \hat{d}(k), \dots, - B_d \hat{d}(k))\in \R^{m_z},
 \end{align}
\end{subequations}
and, defining $\tilde{\vartheta}$ as a trimmed version of $v$ without $(\tx_0,\tu_0)$, i.e., $\tilde{\vartheta} \doteq (\th_0, \tx_1,\tu_1,\th_1, {\dots}, \tx_{N-1}, \tu_{N-1}, \th_{N-1}, \tx_{s}, \tu_{s}, \th_{s})$,
$$g(v) = \delta_{\{\tu_0 \colon \underline{u} \leq \tu_0 \leq \overline{u}\}}(\tu_0) + \gamma_{\beta/2}(\tilde{\vartheta};\overline{\vartheta},\underline{\vartheta}).$$

Step~\ref{alg:ADMM:step_z} in Algorithm~\ref{alg:ADMM} solves the equality-constrained QP problem
\begin{subequations} \label{z_problem}
    \begin{align}
z^{l+1} = \arg \min_{z} & \; \frac{1}{2} z\T P z + (p^l)\T z \label{eq:z_problem:functional}\\
\textrm{s.t.} & \; Gz=b,
    \end{align}
\end{subequations}
where $P = H + \rho D\T D$ and $p^l = q + D\T(\lambda^l - \rho v^l)$. Given that $D$ has full column rank\footnote{The vector $v$   includes $\tx_0$ to ensure that $D$ has full column rank.}, the Hessian $P$ in~\eqref{eq:z_problem:functional} is positive definite. Moreover, since the system~\eqref{eq:lin_sys} is assumed to be controllable, $G$ has full row rank. Consequently, the solution of problem~\eqref{z_problem} is unique.

Solving~\eqref{z_problem} is the most computationally demanding step of Algorithm~\ref{alg:ADMM}. To attain efficiency, the particular semi-banded structure\footnote{See~\citet[Definition~1]{Gracia_MPCT_solver_24}.} of $P$ is exploited. The method works as follows. First, the equality-constrained problem~\eqref{z_problem} can be solved using the explicit solution method described in~\citet[\S 5.5.3]{Boyd_ConvexOptimization}. This method relies on finding a solution of the Karush-Kuhn-Tucker (KKT) optimality conditions, which in this case amounts to solving a linear system of equations. This, in turn, can be decomposed in three simpler linear systems
\begin{subequations} \label{eq:semi_banded_linear_systems}
    \begin{align}
         P \xi &= p^l, \label{eq:1_semi_banded_linear_system}\\
         W \mu^* &= -(G \xi + b), \label{eq:2_semi_banded_linear_system} \\
         Pz^{l+1} &= -(G\T \mu^* + p^l),\label{eq:3_semi_banded_linear_system}
    \end{align}
\end{subequations}
where $W \doteq G P^{-1} G\T \in \R^{m_z \times m_z}$ is also semi-banded, and $\mu^* \in \R^{m_z}$ is the dual optimal solution of problem~\eqref{z_problem}.
The Woodbury matrix identity \citep{woodbury} is leveraged to decompose the semi-banded matrices arising in~\eqref{eq:semi_banded_linear_systems}. In particular, steps~\eqref{eq:1_semi_banded_linear_system} and~\eqref{eq:3_semi_banded_linear_system} are computationally inexpensive, as each of them involves solving two block-diagonal systems separated by a simpler step involving two small-scaled matrix-vector products. In turn, \eqref{eq:2_semi_banded_linear_system} is more demanding due to the structure of~$W$, as instead of the two block-diagonal systems, it requires solving two systems of the form $\tilde{\Gamma} w = c$, where $\tilde{\Gamma}$ has a specific banded-diagonal structure. This structure is exploited in \cite{Krupa_TCST_21}, where the authors use the Cholesky factorization of $\tilde{\Gamma}$ and apply forward-backward substitution; a strategy also adopted in this solver.

Step~\ref{alg:ADMM:step_v} of Algorithm~\ref{alg:ADMM} implies solving the separable optimization problem
\begin{subequations}
\begin{align*}
    v^{l+1} = \arg \min_{v} &\; \frac{\rho}{2} \| Dz^{l+1}-v \|_2^2 - (\lambda^l)\T v + \gamma_{\beta{/}2}(\tilde{\vartheta} ; \overline{\vartheta}, \underline{\vartheta}) \\
    \textrm{s.t.} &\; \underline{u} \leq \tu_0 \leq \overline{u}.
\end{align*}
\end{subequations}
Defining $c_{(j)} \doteq (Dz^{l+1})_{(j)} + \frac{1}{\rho}\lambda^l_{(j)}$, it follows from the definition of $D$ that the components of $v^{l+1}$ related to $\tx_0$, i.e., $v^{l+1}_{(j)}, \, j\in \mathbb{I}_{1}^{n}$, can be computed explicitly as
\begin{equation*}
    v^{l+1}_{(j)} = c_{(j)} = z^{l+1}_{(j)}+\frac{1}{\rho}\lambda^l_{(j)},
\end{equation*}
and the components related to $\tu_0$, i.e., $v^{l+1}_{(j)}, \, j\in \mathbb{I}_{n+1}^{n+m}$, as
\begin{equation*}
    \begin{aligned}
        v^{l+1}_{(j)} &= \min\(\max\(c_{(j)},\underline{u}_{(j-n)}\),\overline{u}_{(j-n)}\) \\
        &= \min\(\max\(z^{l+1}_{(j)}+\frac{1}{\rho}\lambda^l_{(j)},\underline{u}_{(j-n)}\),\overline{u}_{(j-n)}\).
    \end{aligned}
\end{equation*}
Defining $i \doteq j-n-m$, the remaining elements of $v^{l+1}$, i.e., $v^{l+1}_{(j)}, \, j\in \mathbb{I}_{n+m+1}^{n_v}$, can be computed as
 \begin{equation*}
 v^{l+1}_{(j)} = 
\begin{cases}
c_{(j)} {+} \frac{\beta_{(i)}}{2\rho} & \hspace{-0.3em} \text{if} \ c_{(j)} {+} \frac{\beta_{(i)}}{2\rho} \leq \underline{\vartheta}_{(i)}, \\

\underline{\vartheta}_{(i)} & \hspace{-0.3em} \text{if} \ c_{(j)} {+} \frac{\beta_{(i)}}{2\rho} > \underline{\vartheta}_{(i)} \ \text{and} \ c_{(j)} < \underline{\vartheta}_{(i)}, \\

c_{(j)} & \hspace{-0.3em} \text{if} \ \underline{\vartheta}_{(i)} \leq c_{(j)} \leq \overline{\vartheta}_{(i)}, \\

\overline{\vartheta}_{(i)} & \hspace{-0.3em} \text{if} \ c_{(j)} > \overline{\vartheta}_{(i)} \ \text{and} \ c_{(j)} {-} \frac{\beta_{(i)}}{2\rho} < \overline{\vartheta}_{(i)}, \\

c_{(j)} {-} \frac{\beta_{(i)}}{2\rho} & \hspace{-0.3em} \text{if} \ c_{(j)} {-} \frac{\beta_{(i)}}{2\rho} \geq \overline{\vartheta}_{(i)},

\end{cases}
\end{equation*}
thus providing an explicit expression for updating~$v^{l+1}$ thanks to the specific structure of $\gamma_\beta$. As a result, this implementation of MPCT leads to solution times comparable to those of the standard MPC equivalent formulation, that is, without the artificial reference and soft constraints; see \cite{Gracia_MPCT_solver_24, Gracia_LCSS_24} for further details.

To further improve the performance of the algorithm, a simple warm-start procedure that initializes the ADMM variables $v$ and $\lambda$ is employed. In particular, the initial values $v^0$ and $\lambda^0$ at sample time $k$ are taken, respectively, as the optimal solutions $v^*$ and $\lambda^*$ obtained at sample time $k-1$, shifted $n+m+n_{h}$ positions forward. As this method does not determine the last $n+m+n_{h}$ components of $v^0$ and $\lambda^0$, these are simply copied from those of~$v^*$ and~$\lambda^*$ obtained at $k-1$, respectively.

\section{Hardware-in-the-loop results} \label{sec:HiL}

A HiL scheme is set up in which a Siemens S7-1500 PLC implements~\eqref{eq:MPCT_soft} and the estimator~\eqref{eq:observer} to control the nonlinear system~\eqref{eq:non_lin_sys}. The system dynamics are simulated via numerical integration on a PC, where the PLC solve times and communication delays affect the simulation. The PLC and the PC communicate via OPC UA \citep{OPC_UA}, with the former acting as the OPC server and the latter as the client. Figure~\ref{fig:HiL_scheme} shows a scheme of the HiL configuration.

\begin{figure}
\centering
 \includegraphics[width=\linewidth]{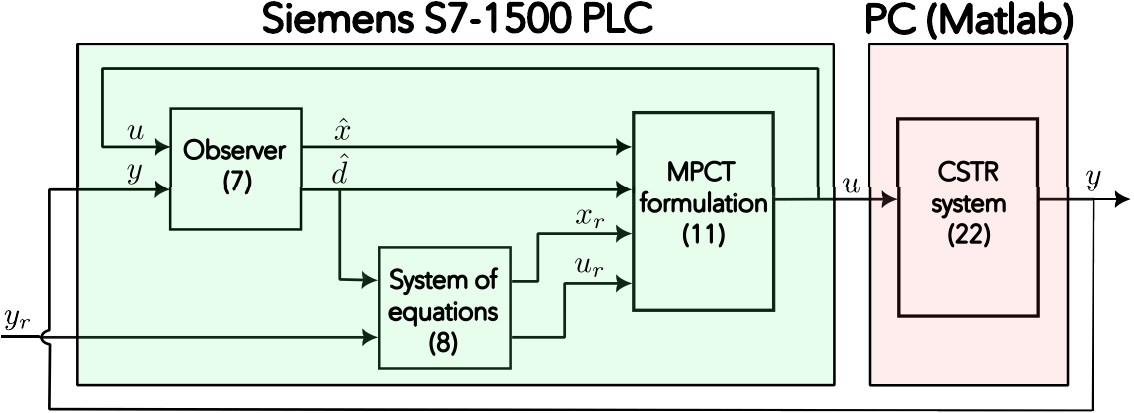}
        \caption{HiL scheme.}
        \label{fig:HiL_scheme}
\end{figure}

The PLC operates in cyclic mode: this is an asynchronous mode of operation, where a new cycle begins immediately after the previous one ends, and its duration depends on the tasks executed within the cycle (e.g., ADMM iterations, monitoring or communication), see~\cite{Krupa_TCST_21}. To avoid excessive consecutive CPU time on the solver task, the ADMM iterations are distributed across multiple PLC cycles, performing a fixed number of iterations in each cycle and verifying the termination criterion at the end. If satisfied, the control input is applied to the system; otherwise, the PLC continues to iterate from its last state.
This approach allows other processes in each cycle to run more frequently rather than waiting for the optimization solver to finish.
Further details on splitting the ADMM iterations on PLCs can be found in~\cite{Krupa_TCST_21}.
 
 The probabilistic validation approach from Section~\ref{sec:prob_validation} is used to select the hyperparameters of the MPCT controller and to validate its long-term operation. Specifically, a set of back-off parameters is selected to tighten the model constraints, and the weights $\beta$ used for the soft constraints. Since constraint satisfaction is considered critical in this case, the aim is to reduce constraint violations while keeping the number of iterations and computation time of the MPC solver low.

\subsection{Testing system: continuous stirred-tank reactor}\label{sec:cstr}

The system considered is a realistic model of a CSTR equipped with a cooling jacket, described in \cite{CSTR_SUBRAMANIAN}. The main reaction, $R_{AB}$, is the transformation of a substance \textit{A} into a product \textit{B}. However, \textit{A} also undergoes an undesirable parallel reaction, $R_{AD}$, producing a product \textit{D}. Additionally, an unwanted consecutive reaction, $R_{BC}$, transforms \textit{B} into a substance \textit{C}. The whole process is summarized as
\begin{equation}
\begin{aligned}
    A &\overset{k_1}\longrightarrow B \overset{k_2}\longrightarrow C,\\
    2A &\overset{k_3}\longrightarrow D,
\end{aligned}
\end{equation}
where $k_i, i \in \I_1^3$, are the respective reaction velocities.
The system is described by the nonlinear differential equations%
\begin{subequations} \label{eq:CSTR_sys}
    \begin{align}
        \dot{c}_A =& \frac{F_I}{V_R} (c_{A0}-c_A) - k_1(\theta) c_A - k_3(\theta) {c_A}^2, \\
        \dot{c}_B =& -\frac{F_I}{V_R} c_B + k_1(\theta) c_A - k_2(\theta)c_B, \\
        \dot{\theta} =& \frac{F_I}{V_R} (\theta_d-\theta) - \frac{1}{\rho C_p}( k_1(\theta) c_A \Delta H_{R_{AB}} {+} k_2(\theta) c_B \Delta H_{R_{BC}} \nonumber\\ & + k_3(\theta){c_A}^2 \Delta H_{R_{AD}} ) + \frac{k_w A_R}{\rho C_p V_R}(\theta_K-\theta), \\
        \dot{\theta}_K =& \frac{1}{m_K C_{pK}}(P_K + k_w A_R (\theta - \theta_K)),
\end{align}
\end{subequations}%
where $c_A$ [mol/l] and $c_B$ [mol/l] are the concentrations of $A$ and $B$, respectively, and $\theta$ [$^\circ \mathrm{C}$] and $\theta_K$ [$^\circ \mathrm{C}$] are the temperatures of the reactor and cooling jacket, respectively. Each velocity $k_i$ depends on the temperature $\theta$ as
\begin{equation}
    k_i(\theta) = k_{i0} \cdot \text{exp}\(\frac{E_i}{\theta + 273.15}\), \ i \in \I_1^3.
\end{equation}
The system has two inputs: the normalized input flow $F_N\doteq\frac{F_I}{V_R}\,[\mathrm{h^{-1}}]$; and the heat removal power $P_K\,\mathrm{[kJ/h]}$ of the cooling jacket.
Two controlled outputs are considered: the concentration $c_B$; and the production rate of $B$, that is, $p_B \doteq c_B F_I \, \mathrm{[mol/h]}$.
Table~\ref{tab:CSTR_parameters} includes a full list of the system parameters, with further details available in~\cite{CSTR_SUBRAMANIAN}.
The constraints considered are $3  \leq F_N \leq 35 $, $-9000 \leq P_K \leq 0$, $\theta \leq 117$, $c_B \geq 0.72$ and $p_B \geq 155$. Note that the constraint on $p_B$ is a coupled state-input constraint. Additionally, the system is subject to disturbances due to fluctuations in the temperature~$\theta_d$ of the input flow~$F_I$, which is modeled according to the discrete-time Singer model 
\begin{equation} \label{eq:Singer_model}
    \theta_d(k+1) = 0.99 \cdot \theta_d(k) + (1-0.99) \cdot \theta_0 + w_d(k),
\end{equation}
where $\theta_0=104.9\,\mathrm{^\circ C}$ is the nominal temperature value and $w_d(k)$ is a white Gaussian noise sampled from a normal distribution with zero mean and variance~$0.01$.

\begin{table}[t]
    \centering
    \small
    \setlength{\tabcolsep}{2pt}
    \begin{tabular}{c|c|c|c}
        Parameter & Symbol & Value & Unit \\ \hline 
        Collision factor of $R_{AB}$ & $k_{10}$ & $1.287{\cdot}10^{12}$ & [1/h]\\
        Collision factor of $R_{BC}$ & $k_{20}$  & $1.287{\cdot}10^{12}$ & [1/h]\\
        Collision factor of $R_{AD}$ & $k_{30}$  & $9.043{\cdot}10^9$ & [1/(mol$\cdot$h)]\\
        Activation energy of $R_{AB}$ & $E_1$ & -9758.3 & [K] \\
        Activation energy of $R_{BC}$ & $E_2$ & -9758.3 & [K] \\
        Activation energy of $R_{AD}$ & $E_3$ & -8560 & [K] \\
        Concentration of $A$ in $F_I$ & $c_{A0}$ & 5.1 & [mol/l] \\
        Temperature of fluid in $F_I$ & $\theta_d$ & 102--108 & [$^\circ $C] \\
        Enthalpy of $R_{AB}$ & $\Delta H_{R_{AB}}$ & 4.2 & [kJ/mol]\\
        Enthalpy of $R_{BC}$ & $\Delta H_{R_{BC}}$ & -11 & [kJ/mol]\\
        Enthalpy of $R_{AD}$ & $\Delta H_{R_{AD}}$ & -41.85 & [kJ/mol]\\
        Density of fluid in $F_I$ & $\rho$ & 0.9342 & [kg/l] \\
        Heat capacity & $C_p$ & 3.01 & [kJ/(kg$\cdot$K)] \\
        Heat transfer coefficient  & \multirow{2}{*}{$k_w$} & \multirow{2}{*}{4.032} & \multirow{2}{*}{[kJ/(hm$^{2}$$\cdot$K)]} \\
        for cooling jacket & & & \\
        Surface of cooling jacket & $A_R$ & 0.215 & [m$^2$] \\
        Reactor volume & $V_R$ & 0.01 & [m$^3$] \\
        Coolant mass & $m_K$ & 5 & [kg] \\
        Heat capacity of coolant & $C_{pK}$ & 2 & [kJ/(kg$\cdot$K)] \\
        Filter time constant of $F_N$ & $T_{F_N}$ & 250 & [s]\\
        Filter time constant of $P_K$ & $T_{P_K}$ & 125 & [s] \\
        \hline
    \end{tabular}
    \caption{Parameters of the CSTR system.}
    \label{tab:CSTR_parameters}
\end{table}

A low-pass filter is applied to the control inputs, defined by
\begin{equation}
    \dot{u}_f(t) = T_f^{-1}(u(t)-u_f(t)), \label{eq:filter}
\end{equation}
where $T_f \doteq \diag\((T_{F_N}, T_{P_K})\)$ are the filter time constants (see Table~\ref{tab:CSTR_parameters}), $u_f \doteq (F_N,P_K) \in \R^m$ are the filtered inputs applied in~\eqref{eq:CSTR_sys}, and $u \in \R^m$ the inputs of the filter, computed by the MPCT solver. Then, the system dynamics are extended by incorporating \eqref{eq:filter} into \eqref{eq:CSTR_sys}, resulting in a plant with state $x \doteq \(c_A, c_B, \theta, \theta_K, u_f\) \in \R^n$, input $u \in \R^m$, and output $y \doteq (c_B, p_B) \in \R^p$. The use of a low-pass input filter is common in real industrial scenarios to avoid sudden changes in the inputs, improving the life expectancy of the actuators. Indeed, most industrial valves integrate a local PID controller for this reason. 

A discrete linear model \eqref{eq:lin_sys} of the system is obtained by linearizing around the equilibrium point
\begin{equation}
    \begin{aligned} \label{eq:eq_point}
       & F_N = 25\, \mathrm{h}^{-1}, \, P_K=-4000\,\mathrm{kJ/h}, \, c_A = 3.161\, \mathrm{mol/l}, \\ 
       & c_B = 0.912 \, \mathrm{mol/l}, \, \theta = 108.53\,^\circ\mathrm{C}, \, \theta_K = 103.91\,^\circ\mathrm{C},
    \end{aligned}
\end{equation} and discretizing with a sample time of $75$ seconds. Finally, the matrices in~\eqref{eq:sys_constraints_y} are taken as $E=C$ and $F = 0$.

\subsection{Controller and estimator parameters} \label{sec:controller_design}
For the practical MPCT formulation~\eqref{eq:MPCT_soft}, the prediction horizon length is taken as $N {=} 7$, and the weight matrices as $Q {=} \diag\((0.1,0.1,5,5,20,30)\)$, $R {=} \diag\((20,30)\)$,
$S = \diag\((10, 50)\)$ and $T = \diag\((0.7,0.7,35,35,10,50)\)$. 
The ADMM parameters are set to $\rho = 40$, $\epsilon_p = 5 \cdot 10^{-3}$ and $\epsilon_d = 1 \cdot 10^{-3}$, with~$\rho$ empirically selected as the one providing the smallest maximum number of iterations in some initial tests.
For better performance of ADMM, the MPC optimization problem is preconditioned by scaling the input and state using the matrices $N_u = \diag\((2,10^{-3})\)$ and $N_x = \diag\((5,20,1,1,2,10^{-3})\)$, respectively, and the coupled constraints~\eqref{eq:sys_constraints_y} using $N_c = \diag\((20,0.5)\)$. For further details on this procedure, see~\cite{Krupa_TCST_21}.

As explained in~\cite{MAEDER}, the only requirement for the design of $L_x$ and $L_d$ is that the estimator~\eqref{eq:observer} is stable. In this case, the controller-observer duality is exploited by applying the LQR design method to~\eqref{eq:observer}, with weights for the LQR stage cost $R'=\diag\((10^3,10^3)\)$ and $Q'=\diag\((1,0.01,1,1,10,10,10^4,10^4)\)$.
Furthermore, as the resulting linear model has no integrator states, $B_d$ is taken as~$0$ to simplify the observer; see~\cite{MAEDER} for further details and guidelines on the design of the estimator for offset-free control.

The next subsection discusses the choice of the weights $\beta$ and the back-off parameters of the controller through the probabilistic validation procedure of Section~\ref{sec:prob_validation}.

\subsection{Validation experiments design} \label{sec:validation_experiments_design}
The selection of hyperparameters in MPC controllers, such as the cost function matrices, is typically made offline through simulations.
However, the impact that some hyperparameters will have on the controller, such as the weights~$\beta$ and back-off parameters, can be more difficult to estimate. Therefore, some parameters are fixed in advance to explore multiple controller configurations depending on the remaining, harder-to-tune parameters.

The objective is to validate the long-term behavior of the closed-loop system for a given selection of MPCT hyperparameters. Specifically, the aim is to verify that the controller performs satisfactorily under the real operating conditions of the plant, where its outputs are required to track a piecewise constant reference.
To this end, each of the $N_s$ validation experiments consists of a closed-loop simulation of the system trajectory over $N_t = 100$ sample times, where:
\begin{itemize}
    \item Before recording each experiment, the closed-loop system is subject to an initialization phase that lasts $40$ sample times using a random reference $y_{r1} \doteq \(c_{Br_{1}}, p_{Br_{1}}\)$, with $c_{Br_{1}}$ and $p_{Br_{1}}$ sampled from uniform distributions within the intervals $[0.73,1.094]$ and $[155,301]$, respectively. Here, the state $x$ starts at the equilibrium point~\eqref{eq:eq_point}, the Singer model~\eqref{eq:Singer_model} at $\theta_0$, and the estimator~\eqref{eq:observer} at the origin. Therefore, the time index $k=0$ denotes the first sample time following the initialization phase, where the closed-loop system is already in steady-state regime.
    \item At a random sample time~$t_r$ drawn from a uniform distribution of integers in the range $[10,50]$, the reference changes to a new value $y_{r2}\doteq \(c_{Br_{2}}, p_{Br_{2}}\)$, sampled from the same probability distribution as~$y_{r1}$.
\end{itemize}
The length of the experiments and the time window for changing the reference were selected so that the experiments could capture both the transient and the steady-state behavior of the closed-loop system. It is observed that the controller design in Section~\ref{sec:controller_design} typically steers the system to steady-state regime in approximately $40$~sample times, which explains the selected length of the initialization phase. Note also that the reference is modeled as a random variable to validate the controller in a wide range of operating regimes.
Referring to the nomenclature used in Section~\ref{sec:prob_validation}, in the experimental setup, a sample $w$ drawn from $\mathcal{W}$ includes: the random reference values $y_{r1}$ and $y_{r2}$; the reference change time instant $t_r$; the initial state $x(0)$; the estimate of the initial state $\hat{x}(0)$ and disturbance $\hat{d}(0)$; the sequence of disturbances $\theta_d(0),\dots,\theta_d(N_t-1)$; and any other stochastic variable affecting the closed-loop system, e.g., the communication times along an experiment.

\begin{table*}[t]
    \centering
    \begin{tabular}{c||c|c|c||c|c|c}
    	& \multicolumn{3}{c||}{\textbf{Controller validation}} 
     	& \multicolumn{3}{c}{\textbf{Empirical verification of~\eqref{eq:probability}}} \\ 
 	\hline
        Hyper-  & Constraint & ADMM & Feasible & Experiments & Experiments & Feasible \\ 
        parameters & violations & iterations & trajectories & where & where & trajectories\\
        $\{\overline{\eta}_\theta,\underline{\eta}_c, \underline{\eta}_p, \beta\}$& $\phi^{1,[5]}$ & $\phi^{2,[5]}$ & ($\phi^1=0$) & $\phi^1 \leq \phi^{1,[5]}$ & $\phi^2 \leq \phi^{2,[5]}$ & ($\phi^1=0$)\\ \hline
        $\{0,0,0,100\}$ & 6.6753 & 177 & 96.9\% & 99.6\% & 99.2\% & 96.3\%\\ \hdashline
        $\{0,0,0,300\}$ & 6.6016 & 637 & 96.9\% & 99.6\% & 99.7\% & 96.3\% \\ \hdashline
        $\{0,0,10,100\}$ & 1.4256 & 240 & 98.1\% &  99.5\% &  99.1\% &  97.6\%\\ \hdashline
        $\{0,0,10,300\}$ & 1.4414 & 736 & 98.1\% & 99.5\% & 99.7\% & 97.6\%\\ \hdashline
        $\textbf{\{1.5,0.08,20,100\}}$ & 0.0045 & 246 & 99.5\% & 99.7\% & 99.5\% & 99.6\% \\ \hdashline
        $\{1.5,0.08,20,300\}$ & 0.0048 & 740 & 99.5\% & 99.7\%  & 99.5\% & 99.6\% \\ \hdashline
        $\{3,0.08,20,100\}$ & 0 & 246 & 99.8\% & 99.9\% & 99.2\% & 99.9\% \\ \hdashline
        $\{3,0.08,20,300\}$ & 0 & 740 & 99.8\% & 99.9\% & 99.5\% & 99.9\% \\ \hline
    \end{tabular}\\
    \caption{On the left: Probabilistic bounds (see~\eqref{eq:probability}) obtained for a subset of the $M = 54$ candidate controllers with $K=2$,  $\delta = 1\cdot10^{-6}$, $r=5$, $\epsilon = 0.03$ and $N_s=1156$ according to~\eqref{eq:samples_condition}. The selected controller~$C_1$ is highlighted in bold. On the right: Empirical verification of the probabilistic bounds.  The percentage of feasible trajectories is to be considered as complementary information (not covered in this case by the result of Proposition~\ref{proposition}).}
    \label{tab:validation_results}
\end{table*}

\begin{remark}
    The proposed validation experiments are taken to resemble the typical operation of the plant, capturing both the transient and steady-state behavior of the closed-loop system. That is, they are structured so that the probabilistic bounds obtained are meaningful for assessing long-term closed-loop performance, either under a constant reference or when it undergoes a step change.
\end{remark}

\subsection{Performance indicators and hyperparameters}

The following performance indicators are considered
\begin{subequations}\label{eq:perf_ind}
    \begin{align}
        \phi^1(w) \doteq & \sum_{k=0}^{N_t-1}\rho_\theta\(\max\(\theta\(k\)-\overline{\theta},0\)\)^2 \nonumber\\ 
        & + \rho_{c}\( \max\(\underline{c}_B-c_B \(k\), 0\)\)^2 \nonumber\\ 
        & + \rho_{p}\(\max\(\underline{p}_B-p_B\(k\), 0\)\)^2, \label{eq:perf_ind_1}\\
        \phi^2(w) \doteq & \max_{k=0,\dots,N_t-1} N_{iter}(k), \label{eq:perf_ind_2}
    \end{align}
\end{subequations}
where $\overline{\theta} = 117$, $\underline{c}_B = 0.72$ and $\underline{p}_B = 155$ are the bounds of the constrained variables as defined in Section~\ref{sec:cstr}. The positive scalars $\rho_\theta$, $\rho_c$ and $\rho_p$ normalize and weight the violation of the constraints, and $N_{iter}(k)$ denotes the number of iterations performed by the ADMM algorithm at time~$k$.
Parameters are set to $\rho_\theta = 30$, $\rho_c = 150$ and $\rho_p = 1$.

The performance indicator $\phi^1$ defined in~\eqref{eq:perf_ind_1} gives a measure of constraint violations during a closed-loop experiment. Constraint violations of the inputs are not penalized, as the applied inputs are hard constrained by~\eqref{eq:MPCT_soft:hard_u0} and fit the form required in Section~\ref{sec:cstr}. The $\max(\cdot)$ terms in $\phi^1$ are squared to strongly penalize cases in which the variables significantly exceed their bounds. Note that $\phi^1(\cdot)$ is equal to~$0$ if and only if the closed-loop trajectories are feasible throughout the entire simulation. The indicator~$\phi^2$ in~\eqref{eq:perf_ind_2} returns the maximum number of ADMM iterations throughout an experiment. As the number of ADMM iterations is directly related to the computation time of the solver, $\phi^2$ is used to validate that the control action $u(k)$ is computed in a time that is significantly shorter than the sample time of the system.

The proposed probabilistic validation framework is applied to select the values of the back-off parameters and the weights~$\beta$ of the soft constraints. To this end, a preliminary empirical analysis of their effect on the system via several closed-loop simulations is first conducted, enabling the identification of a sensible range of parameter values. Next, based on this range, a set of $M$ candidate controllers is defined, determined by all possible combinations of the hyperparameter values $\overline{\eta}_\theta \in \{0,1.5,3\}$,  $\underline{\eta}_c \in \{0,0.04,0.08\}$, $\underline{\eta}_p \in \{0,10,20\}$ and $\beta \in \{\vv{1}_{n_\vartheta} \cdot 100,\vv{1}_{n_\vartheta} \cdot 300\}$, where $\overline{\eta}_\theta$, $\underline{\eta}_c$ and $\underline{\eta}_p$ tighten the bounds $\overline{\theta}$, $\underline{c}_B$ and $\underline{p}_B$, respectively.
This results in a total of $M = 54$ candidate controllers. As in this particular case only vectors~$\beta$ with identical elements are considered, given a scalar $c\in\R$, $\beta=c$ will stand for $\beta=\vv{1}_{n_\vartheta} \cdot c$ in what follows.

Setting $\delta = 1\cdot10^{-6}$, $\epsilon = 0.03$ and $r=5$, and applying~\eqref{eq:samples_condition} with $K=2$, it is concluded that performing $N_s \geq 1156$ closed-loop experiments is sufficient to individually bound the probabilities defined in~\eqref{eq:probability} with a confidence of at least $1-\delta$. Accordingly, $N_s=1156$ is selected.

\begin{figure*}[t]
    \centering
    \begin{subfigure}{0.33\textwidth}
        \includegraphics[width=\linewidth, height = 2.8cm]{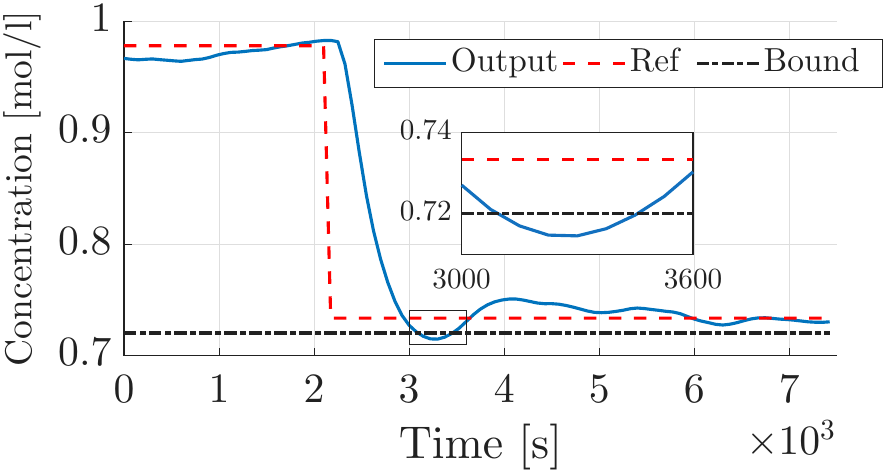}
        \caption{Output $c_B$.}
        \label{subfig:exp:output_1}
    \end{subfigure}%
    \hfill
    \begin{subfigure}{0.33\textwidth}
        \includegraphics[width=\linewidth, height = 2.8cm]{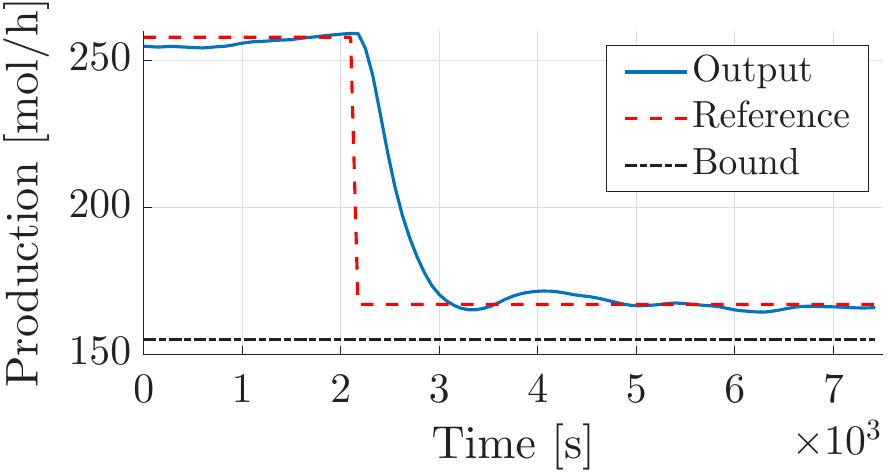}
        \caption{Output $P_B$.}
        \label{subfig:exp:output_2}
    \end{subfigure}%
    \hfill
    \begin{subfigure}{0.33\textwidth}
        \includegraphics[width=\linewidth, height = 2.8cm]{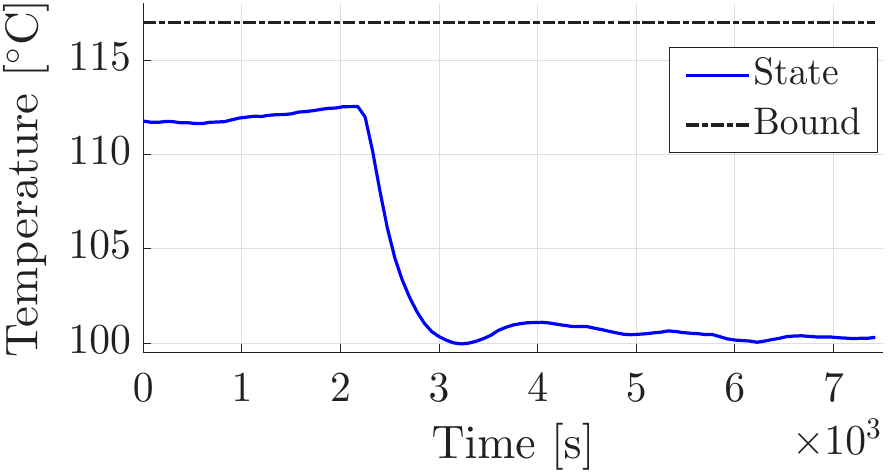}
        \caption{State $\theta$.}
        \label{subfig:exp:state_3}
    \end{subfigure}%
    \hfill
    \vskip\baselineskip
    \centering
    \begin{subfigure}{0.33\textwidth}
        \includegraphics[width=\linewidth, height = 2.8cm]{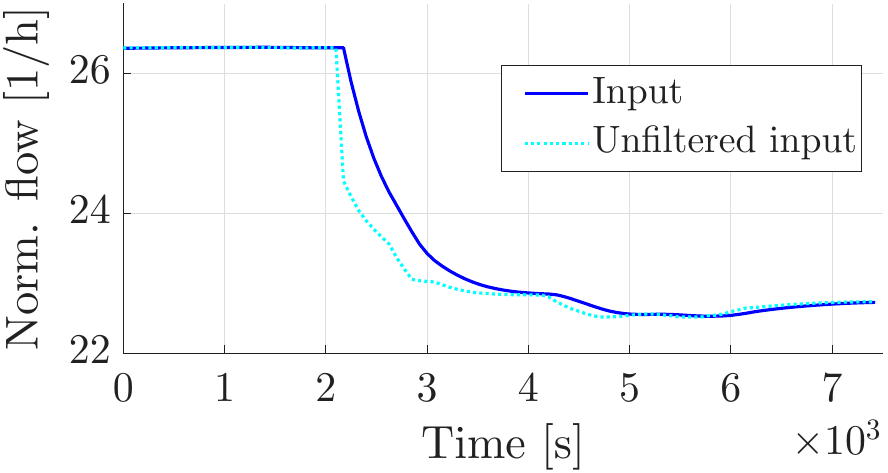}
        \caption{Input $F_N$.}
        \label{subfig:exp:input_1}
    \end{subfigure}%
    \hfill
    \begin{subfigure}{0.33\textwidth}
        \includegraphics[width=\linewidth, height = 2.8cm]{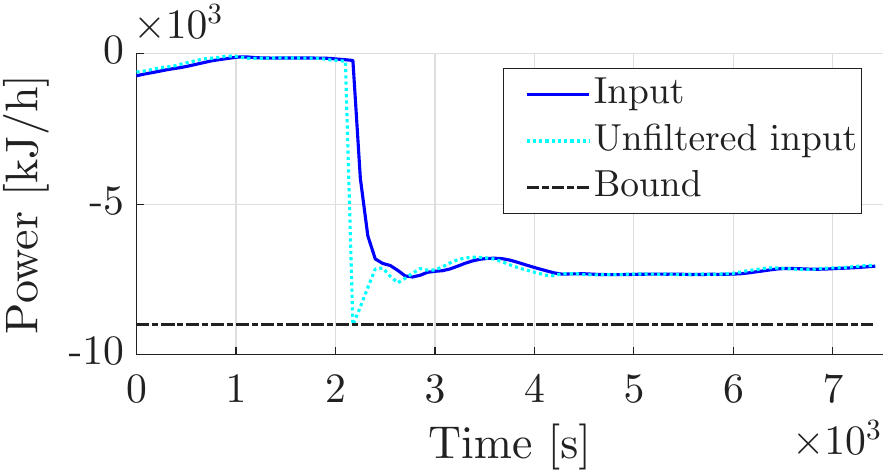}
        \caption{Input $P_K$.}
        \label{subfig:exp:input_2}
    \end{subfigure}%
    \hfill
    \begin{subfigure}{0.33\textwidth}
        \includegraphics[width=\linewidth, height = 2.8cm]{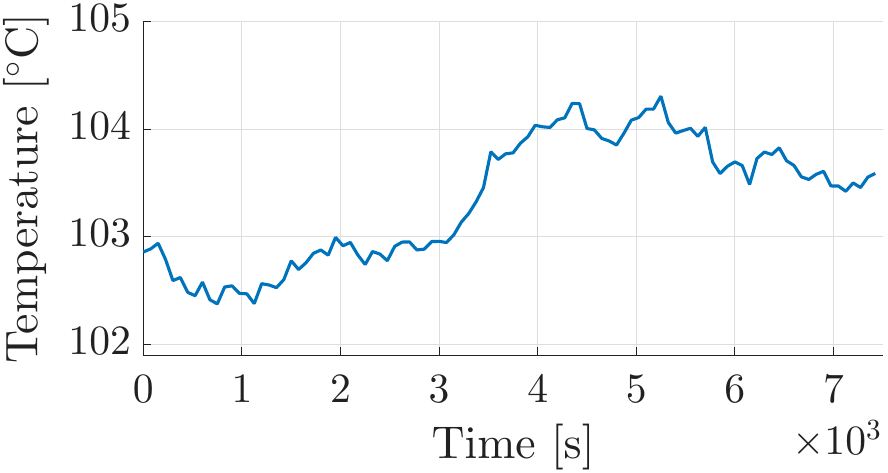}
        \caption{Disturbance $\theta_d$.}
        \label{subfig:exp:disturbance}
    \end{subfigure}%
    \caption{Closed-loop simulation of the CSTR controlled with~$C_1$. This experiment corresponds to the $5$-th worst case of $C_1$, where $\phi^{1}(\cdot) = 0.0045$ (see Table~\ref{tab:validation_results}).}
    \label{fig:experiment}
\end{figure*}

\subsection{Validation experiments}

Proposition~\ref{proposition} does not require evaluating all the $M$ candidate controllers for~\eqref{eq:probability} to hold individually for each controller. In other words, \eqref{eq:probability} is valid for the controller setup~$i$, as soon as $\phi_{i}^{\ell,[r]}$ is experimentally evaluated. Therefore, once an acceptable controller is found, the performance of other candidate setups does not need to be assessed~\citep{karg_prob_validation}.

The left part of Table~\ref{tab:validation_results} shows the performance of the $r$-th worst case, as measured by the indicators in~\eqref{eq:perf_ind}, obtained for a subset of the $M$ controllers. 
Column ``feasible trajectories'' displays the percentage of experiments (out of $N_s$) that resulted in no constraint violations.
The results indicate that, as the back-off parameters $\overline{\eta}_\theta$, $\underline{\eta}_c$ and $\underline{\eta}_p$ increase, the number of experiments resulting in feasible trajectories grows. However, increasing $\beta$ does not appear to significantly improve constraint satisfaction, while the number of ADMM iterations is negatively affected.
Another point to consider is that, even though controllers with larger back-off may seem to perform better according to Table~\ref{tab:validation_results}, excessive back-offs can often make the reference $(x_r(k),u_r(k))$ provided by $\eqref{eq:ref_sys}$ not admissible, and make the controller behave too conservatively. Since constraint satisfaction is considered to be crucial in this case study, a set of back-off parameters is selected to provide low levels of constraint violations. Specifically, the controller with parameters $\{\overline{\eta}_\theta,\underline{\eta}_c, \underline\eta_p, \beta\} = \{1.5,0.08,20,100\}$ is selected, which is denoted as~$C_1$. For comparison, $C_0$ denotes the controller with $\beta=100$ and no back-off, i.e., $\{\overline{\eta}_\theta,\underline{\eta}_c, \underline\eta_p, \beta\} = \{0,0,0,100\}$. From Table~\ref{tab:validation_results} it is concluded that in a new experiment, with high confidence $1-\delta$, the closed-loop system will violate the constraints by no more than $0.0045$, as measured by $\phi^1$ defined in~\eqref{eq:perf_ind_1}, with a probability of at least $100 \cdot(1-\epsilon) = 97\%$. With the same confidence and probability bound, the maximum number of ADMM iterations in a new experiment will be at most~$246$.
Figure~\ref{fig:experiment} shows the closed-loop experiment resulting in the $5$-th worst-case performance according to~$\phi^1$.

\begin{figure*}[t]
    \centering
    \begin{subfigure}{0.33\textwidth}
        \includegraphics[width=\linewidth, height = 2.8cm]{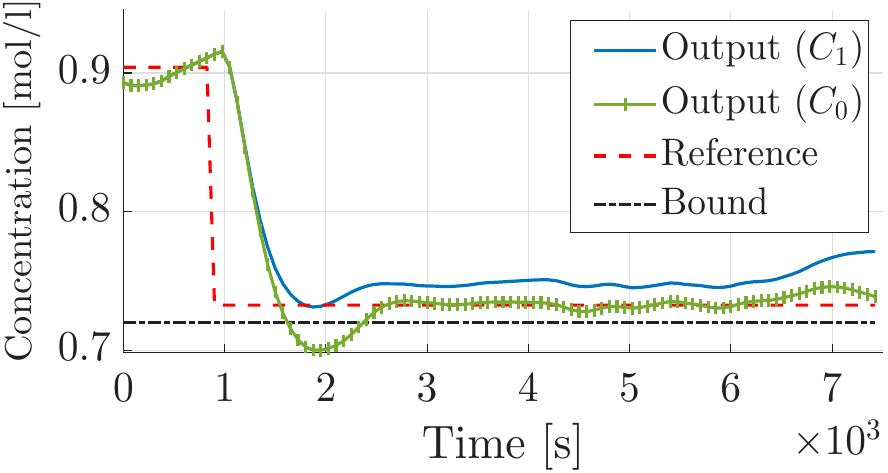}
        \caption{Output $c_B$.}
        \label{subfig:exp_2:output_1}
    \end{subfigure}%
    \hfill
    \begin{subfigure}{0.33\textwidth}
        \includegraphics[width=\linewidth, height = 2.8cm]{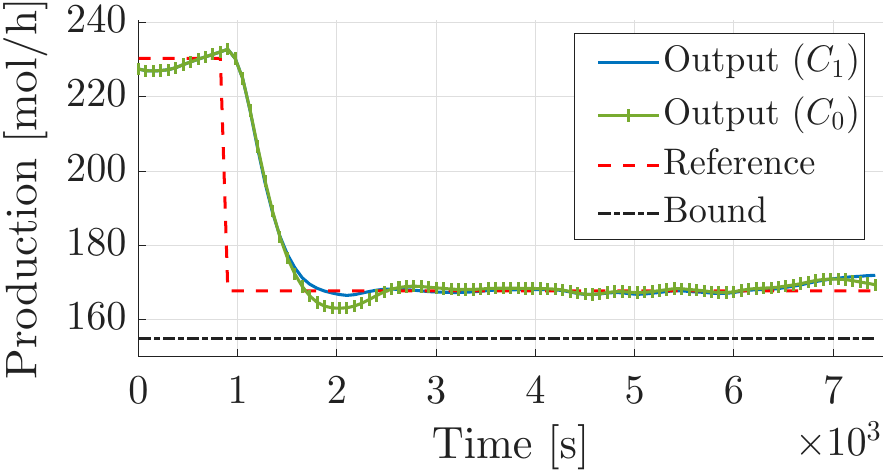}
        \caption{Output $P_B$.}
        \label{subfig:exp_2:output_2}
    \end{subfigure}%
    \hfill
    \begin{subfigure}{0.33\textwidth}
        \includegraphics[width=\linewidth, height = 2.8cm]{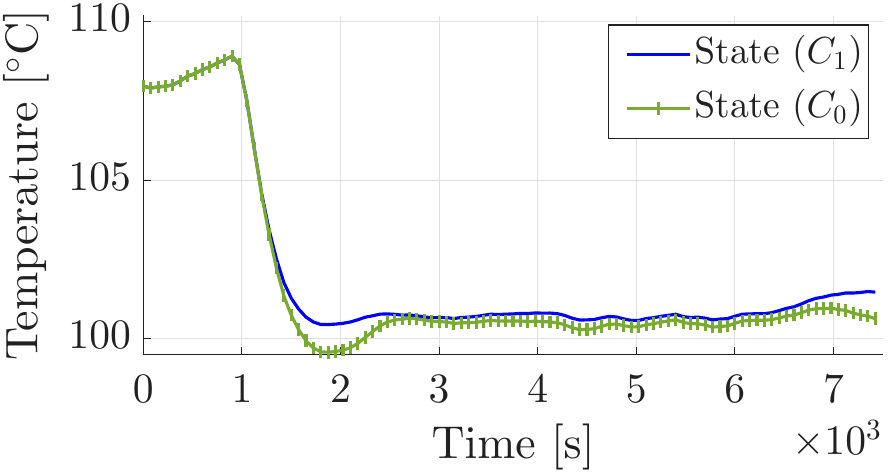}
        \caption{State $\theta$.}
        \label{subfig:exp_2:state_3}
    \end{subfigure}%
    \hfill
    \vskip\baselineskip
    \centering
    \begin{subfigure}{0.33\textwidth}
        \includegraphics[width=\linewidth, height = 2.8cm]{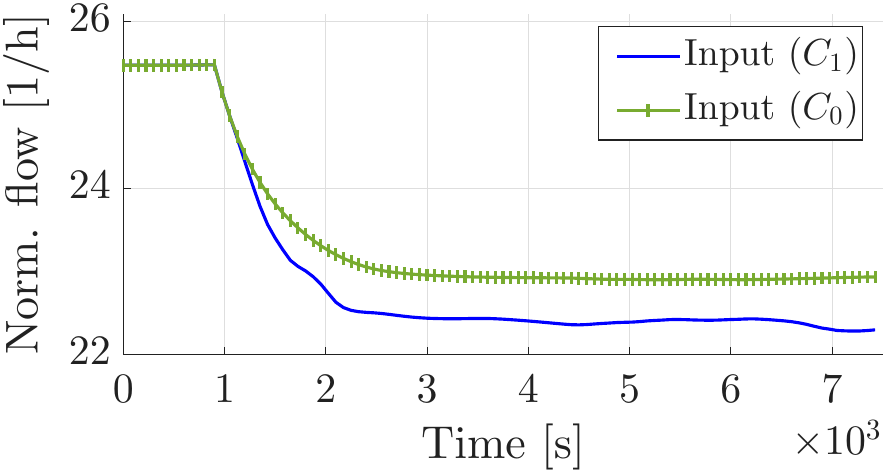}
        \caption{Input $F_N$.}
        \label{subfig:exp_22:input_1}
    \end{subfigure}%
    \hfill
    \begin{subfigure}{0.33\textwidth}
        \includegraphics[width=\linewidth, height = 2.8cm]{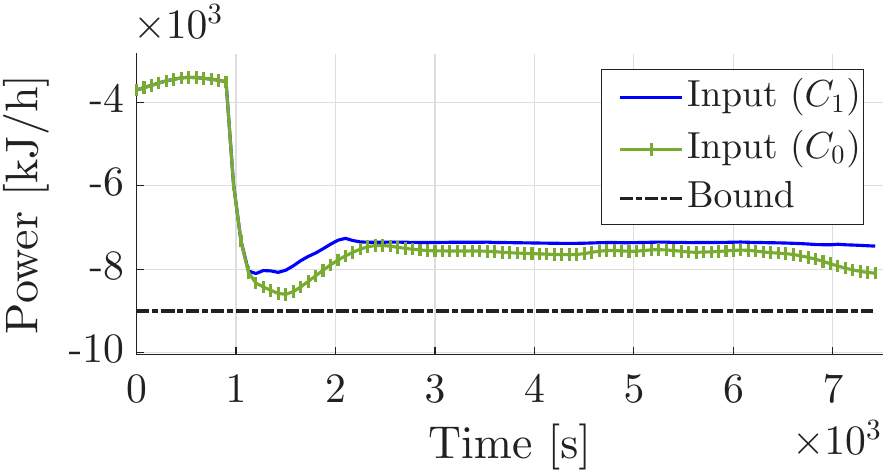}
        \caption{Input $P_K$.}
        \label{subfig:exp_2:input_2}
    \end{subfigure}%
    \hfill
    \begin{subfigure}{0.33\textwidth}
        \includegraphics[width=\linewidth, height = 2.8cm]{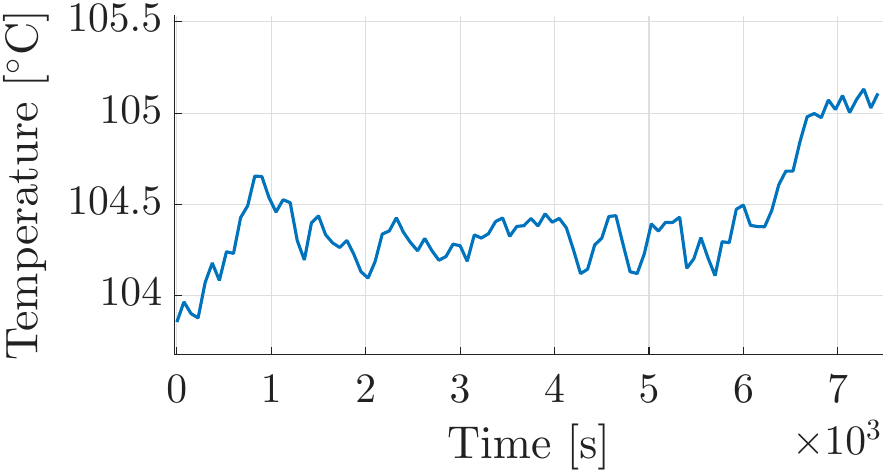}
        \caption{Disturbance $\theta_d$.}
        \label{subfig:exp_2:disturbance}
    \end{subfigure}%
    \caption{Comparison of closed-loop simulations of the CSTR system using the controllers~$C_1$ and $C_0$. A case with nonadmissible reference pairs $(x_r(k),u_r(k))$ due to back-off. }
    \label{fig:experiment_2}
\end{figure*}

\begin{figure*}[t]
    \centering
    \begin{subfigure}{0.25\textwidth}
        \includegraphics[width=\linewidth,height=2.5cm]{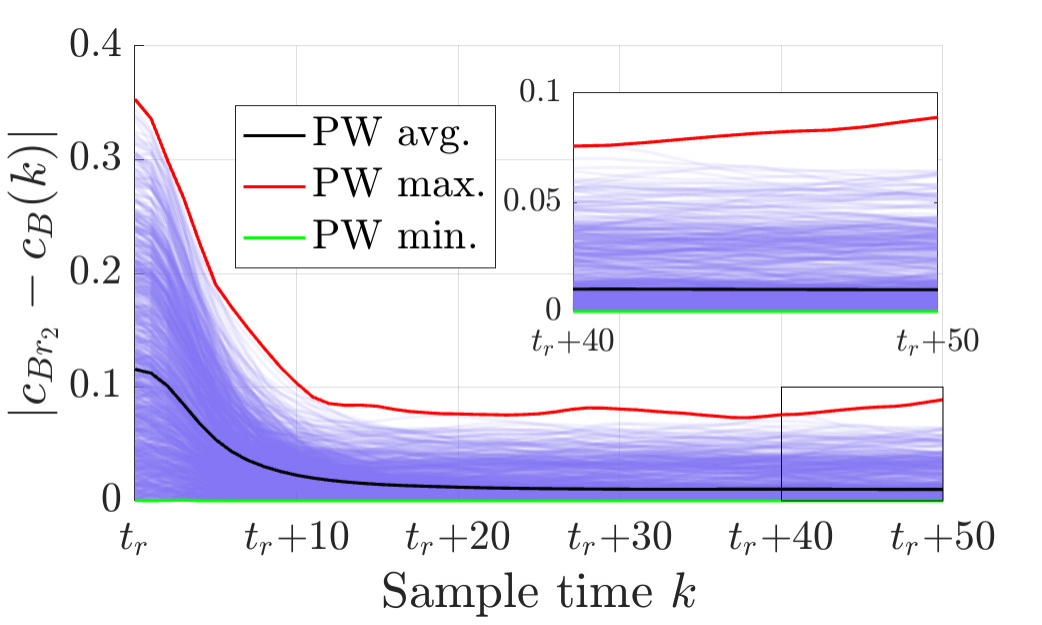}
        \caption{Distance to $c_{Br_{2}}$ using~$C_1$.}
        \label{subfig:dist_to_yr1_back-off}
    \end{subfigure}%
    \hfill
    \centering
    \begin{subfigure}{0.25\textwidth}
        \includegraphics[width=\linewidth,height=2.5cm]{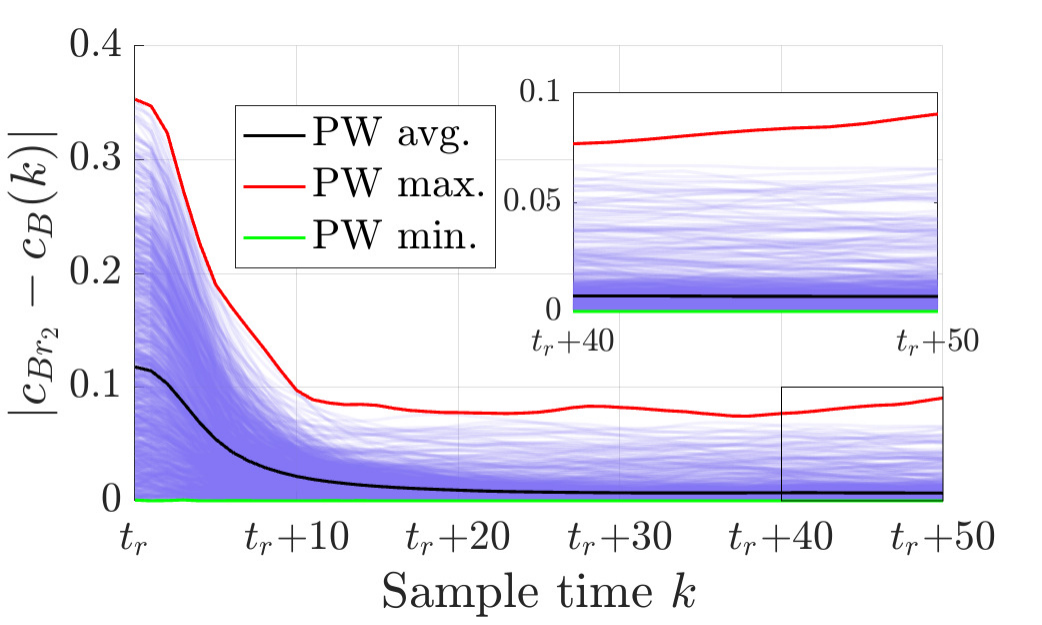}
        \caption{Distance to $c_{Br_{2}}$ using~$C_0$.}
        \label{subfig:dist_to_yr1_no_back-off}
    \end{subfigure}%
    \hfill
    \centering
    \begin{subfigure}{0.25\textwidth}
        \includegraphics[width=\linewidth,height=2.5cm]{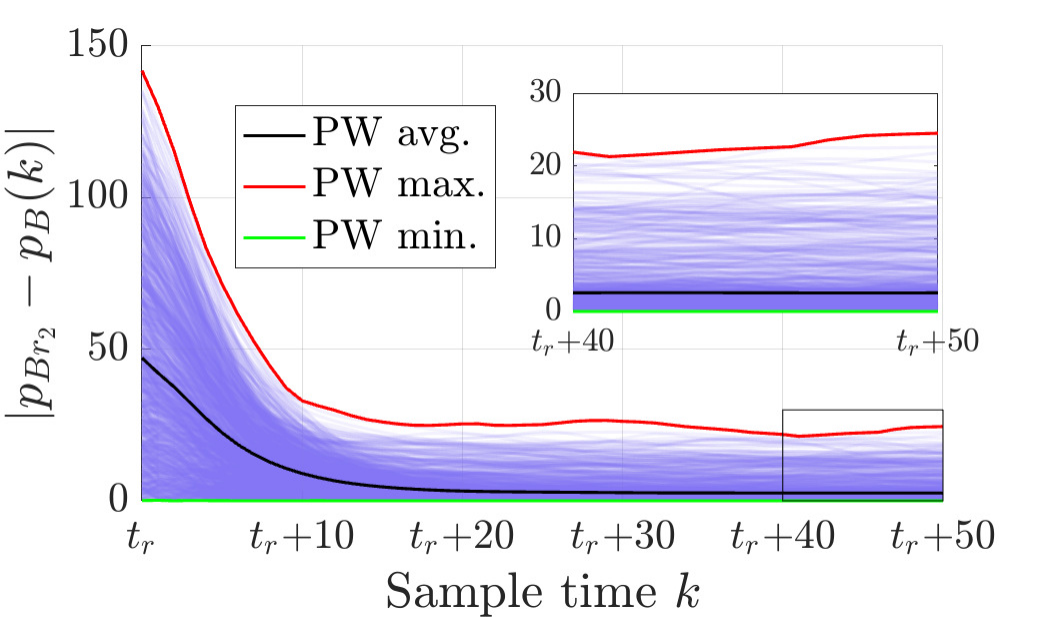}
        \caption{Distance to $p_{Br_{2}}$ using~$C_1$.}
        \label{subfig:dist_to_yr2_back-off}
    \end{subfigure}%
    \hfill
    \centering
    \begin{subfigure}{0.25\textwidth}
        \includegraphics[width=\linewidth,height=2.5cm]{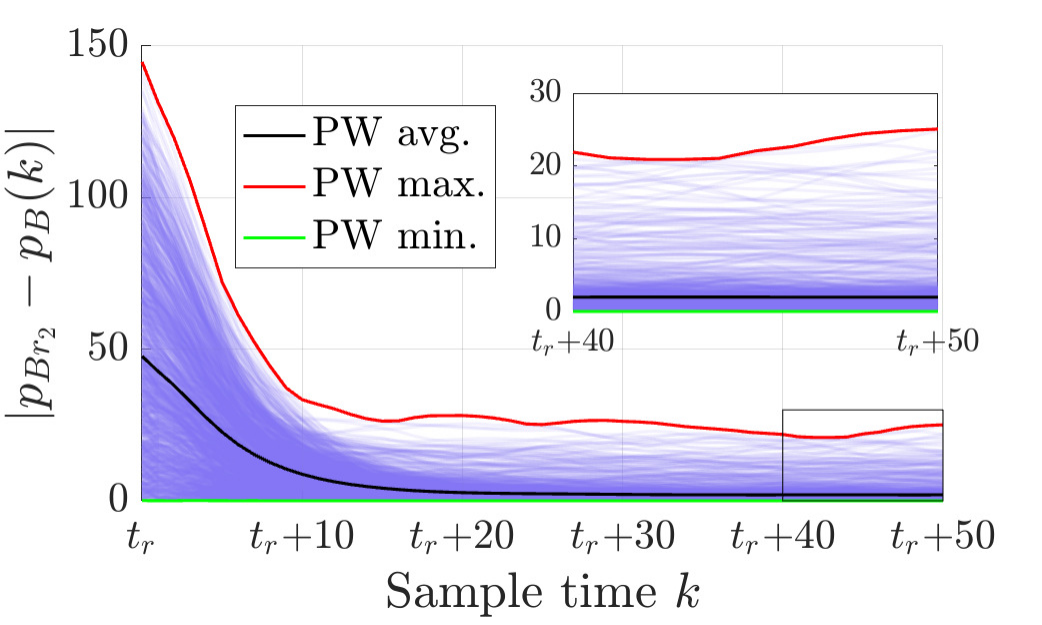}
        \caption{Distance to $p_{Br_{2}}$ using~$C_0$.}
        \label{subfig:dist_to_yr2_no_back-off}
    \end{subfigure}%
    \caption{Comparison of distances from outputs to reference $y_r$ using controllers $C_1$ and $C_0$. In the legend, ``PW'' stands for ``point-wise''.}
    \label{fig:dist_to_yr}
\end{figure*}

\begin{figure*}[h!]
    \centering
    \begin{subfigure}{0.25\textwidth}
        \includegraphics[width=\linewidth,height=2.5cm]{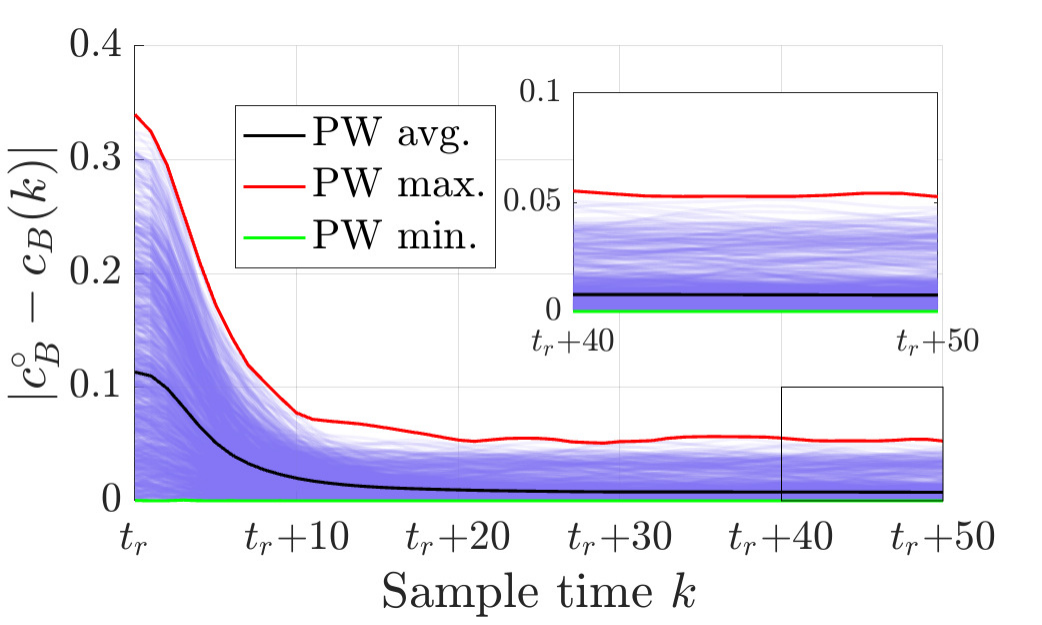}
        \caption{Distance to $c_B^\circ$ using~$C_1$.}
        \label{subfig:dist_to_yo1_back-off}
    \end{subfigure}%
    \hfill
    \centering
    \begin{subfigure}{0.25\textwidth}
        \includegraphics[width=\linewidth,height=2.5cm]{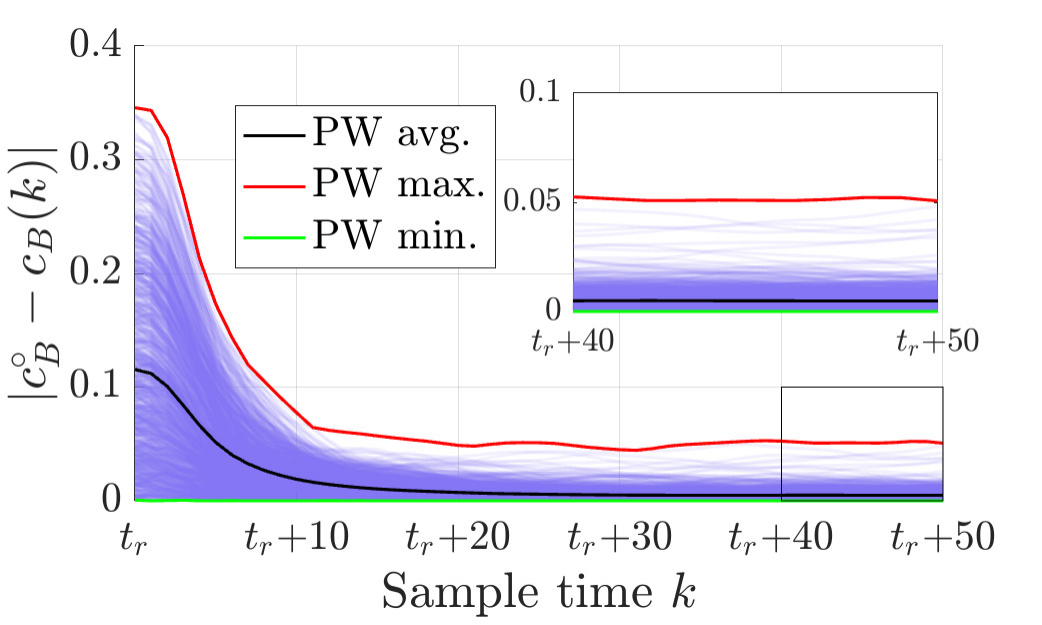}
        \caption{Distance to $c_B^\circ$ using~$C_0$.}
        \label{subfig:dist_to_yo1_no_back-off}
    \end{subfigure}%
    \hfill
    \centering
    \begin{subfigure}{0.25\textwidth}
        \includegraphics[width=\linewidth,height=2.5cm]{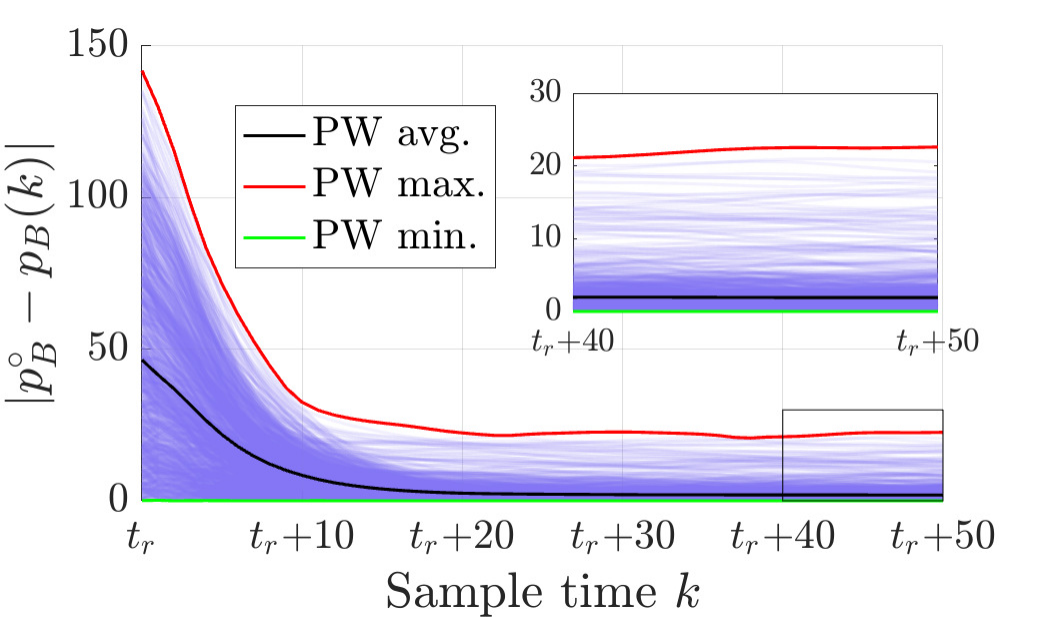}
        \caption{Distance to $p_B^\circ$ using~$C_1$.}
        \label{subfig:dist_to_yo2_back-off}
    \end{subfigure}%
    \hfill
    \centering
    \begin{subfigure}{0.25\textwidth}
        \includegraphics[width=\linewidth,height=2.5cm]{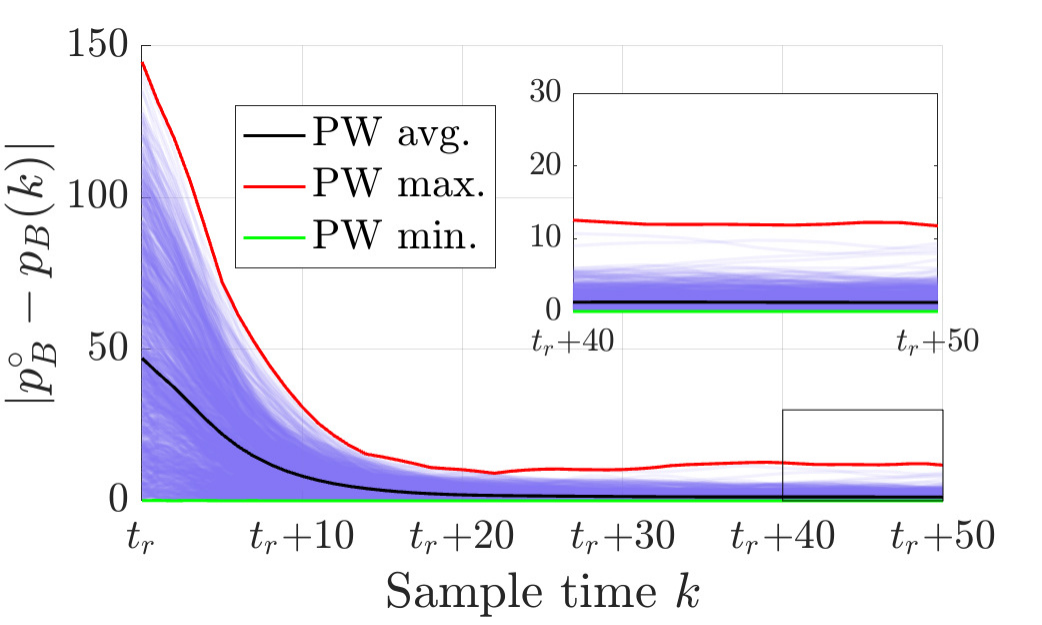}
        \caption{Distance to $p_B^\circ$ using~$C_0$.}
        \label{subfig:dist_to_yo2_no_back-off}
    \end{subfigure}%
    \\
    \caption{Comparison of distances from the outputs to the admissible reference $y^\circ$ using controllers $C_1$ and $C_0$.}
    \label{fig:dist_to_yo_back-off}
\end{figure*}

\begin{figure*}[t]
    \centering
    \begin{subfigure}{0.25\textwidth}
        \includegraphics[width=\linewidth,height=2.49cm]{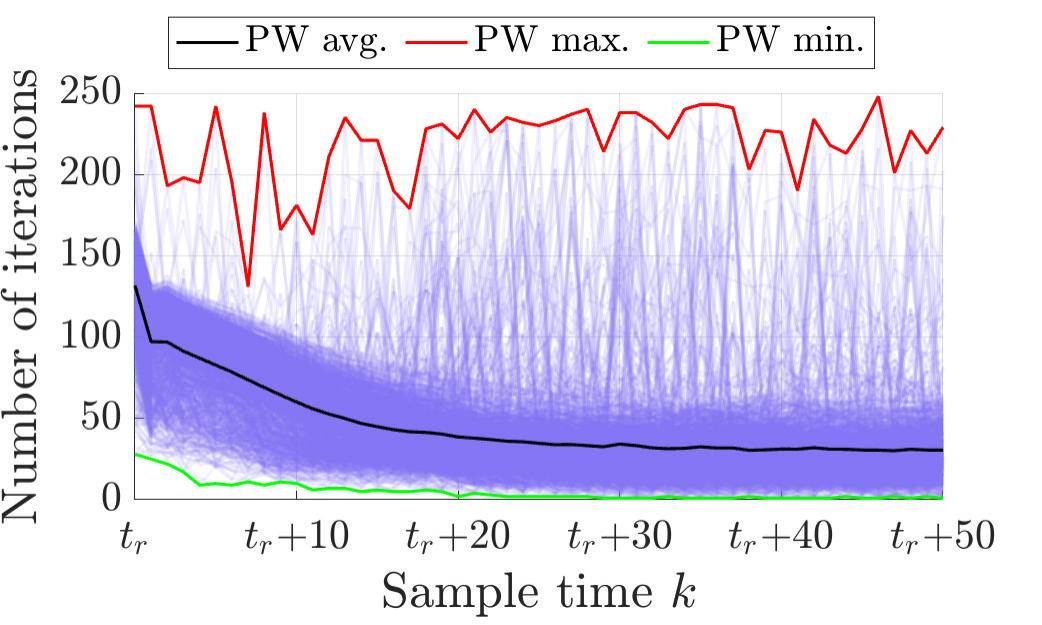}
        \caption{Number of iterations using~$C_1$.}
        \label{subfig:iter_back-off}
    \end{subfigure}%
    \hfill
    \hfill
    \centering
    \begin{subfigure}{0.25\textwidth}
        \includegraphics[width=\linewidth,height=2.49cm]{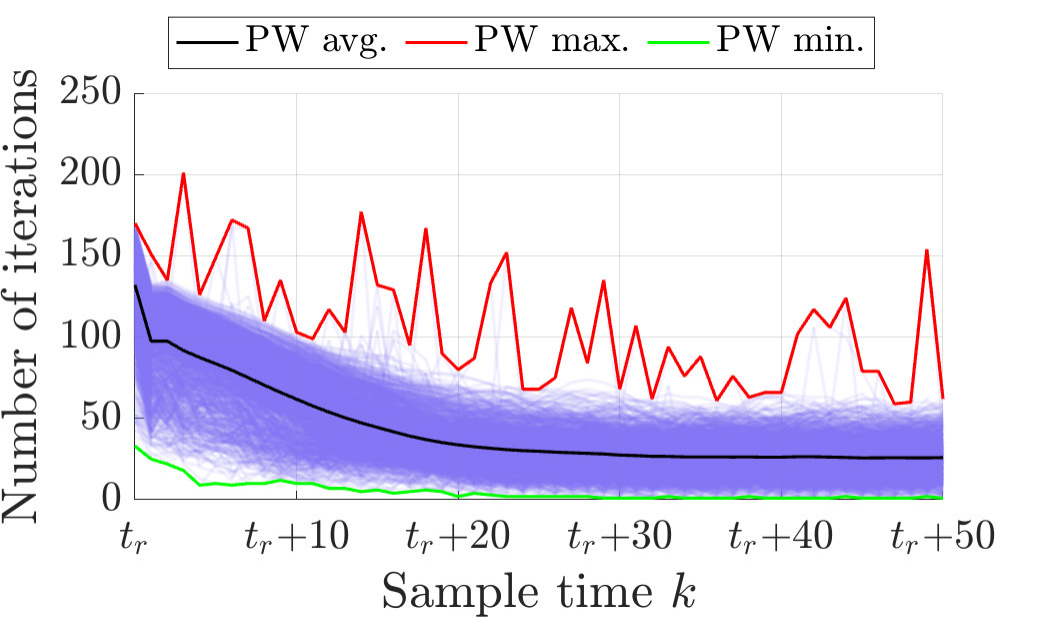}
        \caption{Number of iterations using~$C_0$.}
        \label{subfig:iter_no_back-off}
    \end{subfigure}%
    \hfill
    \centering
    \begin{subfigure}{0.25\textwidth}
        \includegraphics[width=\linewidth,height=2.49cm]{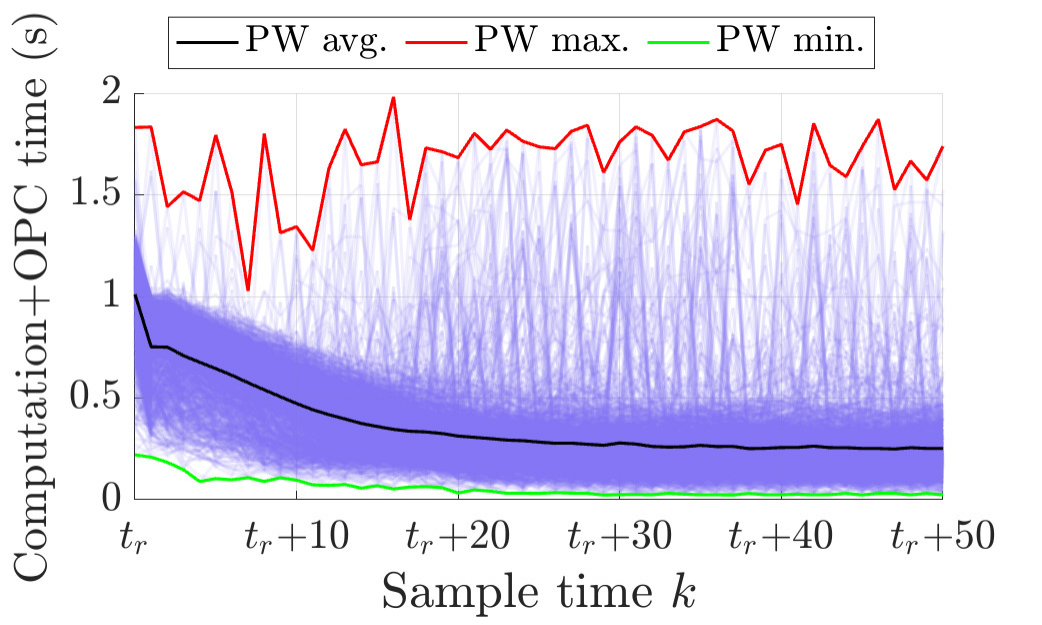}
        \caption{Input availability time using~$C_1$.}
        \label{subfig:time_back-off}
    \end{subfigure}%
    \centering
    \begin{subfigure}{0.25\textwidth}
        \includegraphics[width=\linewidth,height=2.49cm]{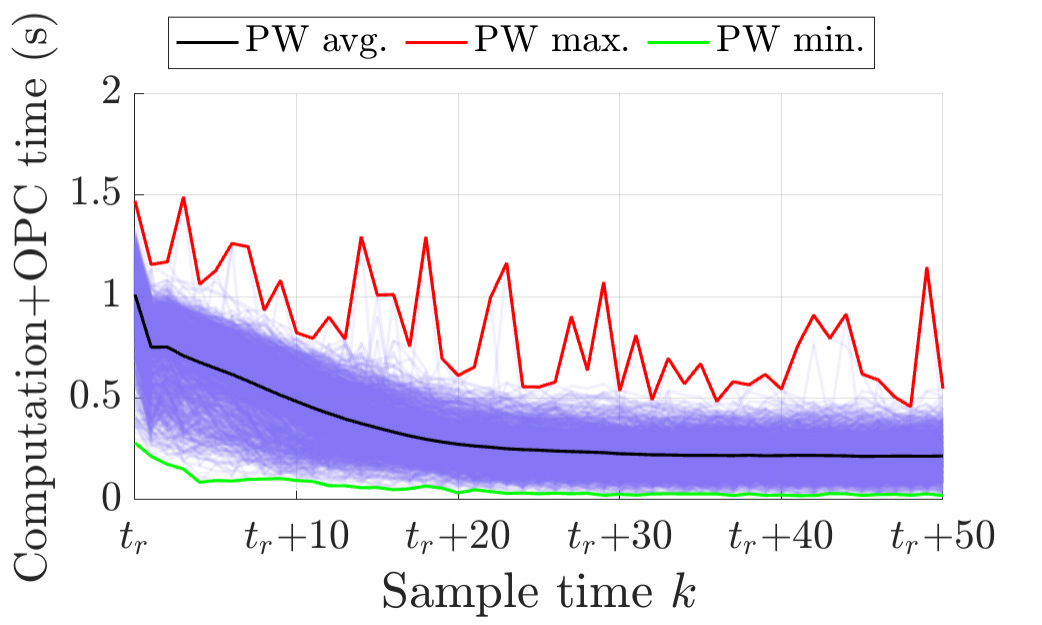}
        \caption{Input availability time using~$C_0$.}
        \label{subfig:time_no_back-off}
    \end{subfigure}%
    \\
    \caption{Comparison of number of iterations and time required to obtain the input using controllers $C_1$ and $C_0$.}
    \label{fig:time_iter}
\end{figure*}

To empirically validate the probabilistic bounds obtained, each controller was tested in $1000$ new \text{i.i.d.} closed-loop experiments. The right part of Table~\ref{tab:validation_results} shows the results, which confirm the probabilistic bounds and indicate that the number of experiments in which the controllers exceed their thresholds is rather small. It was also observed that, in most instances, the probabilistic bounds for both~$\phi^1$ and $\phi^2$ were simultaneously satisfied.

To illustrate that constraint tightening can exacerbate the issue of generating infeasible references $(x_r(k),u_r(k))$ in~\eqref{eq:ref_sys}, thus preventing the system from achieving offset-free tracking to constant output references~$y_r$, Figure~\ref{fig:experiment_2} compares the controllers~$C_1$ and~$C_0$ in one of the $N_s$ PLC validation experiments. The results show that, in this particular case, offset-free tracking of the output~$c_B$ is not achieved with the selected controller~$C_1$, whereas the controller~$C_0$ does achieve offset-free tracking (albeit at the expense of constraint violations).

To assess reference tracking performance of~$C_1$ over the~$N_s$ experiments, the distance between each output and the reference~$y_{r2}$ is recorded from the sample time~$t_r$, in which the reference changes, until $t_r+50$.
In addition, the process is repeated for~$C_0$ to illustrate how back-off affects the ability of the controller to reach the output reference.
The results are shown in Figure~\ref{fig:dist_to_yr}. Comparing Figure~\ref{subfig:dist_to_yr1_back-off} with~\ref{subfig:dist_to_yr1_no_back-off}, and~\ref{subfig:dist_to_yr2_back-off}  with~\ref{subfig:dist_to_yr2_no_back-off}, it is observed that tracking performance improves when using $C_0$ compared to $C_1$ (see the zoomed-in parts), as with the former, offset is reduced in general and a notably larger number of experiments achieve offset-free tracking, of course at the expense of higher constraint violations. However, in many experiments, the reference could not be reached with~$C_0$, partly because the considered output references~$y_{r2}$ were not admissible due to the real system constraints.
Indeed, this situation is not unusual in practice, where the operator may not know if the selected reference is admissible, but still aims to stabilize the system as close as possible to it.

To better assess the effect of back-off and the actual tracking performance of the controllers, the distance to the admissible steady output $y^\circ\doteq\(C_{B}^\circ, p_B^\circ\)$ is computed. This output is obtained in two steps: \emph{(i)} given~$y_{r2}$, compute the steady state-input pair $(x_{u},u_{u})$ to which the closed-loop system (with $\theta_d=\theta_0$) would converge in the absence of constraints; and \emph{(ii)} considering the system constraints~\eqref{eq:non_lin_sys_constraints}, determine the admissible steady output~$y^\circ$ corresponding to the steady state-input $(x_c,u_c)$ obtained by minimizing the distance between $u_c$ and $u_{u}$. The results are illustrated in Figure~\ref{fig:dist_to_yo_back-off}, where it can be noticed that, although back-off helps enforce the constraints, in many cases it prevents the system from approaching $y^\circ$, in contrast to most cases where no back-off is used.

Figure~\ref{fig:time_iter} shows a comparison between the controllers~$C_1$ and $C_0$ in terms of number of iterations and computation-plus-communication time obtained for the $N_s$ experiments between the sample times~$t_r$ and $t_r+50$. Note that the controller~$C_0$ requires in general a lower number of iterations due to a reduced occurrence of active constraints. Nonetheless, the controller $C_1$~provides rather acceptable results while limiting constraint violations to small values, yielding a maximum time on the PLC of $1.98$ seconds (considerably shorter than the sample time of $75$ seconds).
In terms of memory efficiency, the implementation of the estimator~\eqref{eq:observer}, the solver for~\eqref{eq:MPCT_soft} and the reference calculator~\eqref{eq:ref_sys} required $0.73$ MB of memory in the PLC, which constitutes only $3\%$ of the $24$ MB available.

\section{Conclusions} \label{sec:conclusions}

This article has demonstrated the viability of controlling a real (nonlinear) system using a refined linear MPC controller implemented on industrial embedded hardware. Combined with an offset-free scheme, it proved to be a practical solution thanks to several key features: an artificial reference that avoids the need to provide an admissible reference; soft constraints that ensure feasibility of the MPC optimization problem; and back-off parameters to adjust the model constraints to satisfy the actual system constraints. In addition, a probabilistic framework that enables the selection of critical controller parameters and the validation of the long-term closed-loop system operation was proposed. By means of a structure-exploiting ADMM-based solver, it was shown how the formulation can be implemented efficiently, according to embedded systems capabilities.
For the numerical experiments, the formulation was implemented on a Siemens S7-1500 PLC to control a nonlinear CSTR system in a HiL setup. The results confirmed the effectiveness of both the formulation and the probabilistic framework, obtaining a controller with probabilistic bounds on critical aspects such as constraint violations and computational efficiency, thereby meeting the requirements of real-world operation. Overall, this work contributes to bridging the gap between academic research and industrial practice.

\section*{Declaration of competing interest}

The authors declare that they have no known competing financial interests or personal relationships that could have appeared to influence the work reported in this paper.

\section*{Acknowledgments}

\noindent This work was supported from grants PID2022-142946NA-I00 and PID2022-141159OB-I00, funded by MICIU/AEI/ 10.13039/501100011033 and by ERDF/EU.

\noindent V.~Gracia acknowledges support from grant PREP2022-000136 funded by MICIU/AEI/ 10.13039/501100011033 and by ESF+.

\sloppy
\noindent F.~Fele also gratefully acknowledges support from~grant RYC2021-033960-I funded by MICIU/AEI/ 10.13039/ 501100011033 and European Union NextGenerationEU/PRTR.


\setlength{\bibsep}{0.0pt}

\bibliographystyle{elsarticle-harv}
\bibliography{bib_reviewed}

\begin{thebibliography}{37}
\expandafter\ifx\csname natexlab\endcsname\relax\def\natexlab#1{#1}\fi
\providecommand{\url}[1]{\texttt{#1}}
\providecommand{\href}[2]{#2}
\providecommand{\path}[1]{#1}
\providecommand{\DOIprefix}{doi:}
\providecommand{\ArXivprefix}{arXiv:}
\providecommand{\URLprefix}{URL: }
\providecommand{\Pubmedprefix}{pmid:}
\providecommand{\doi}[1]{\href{http://dx.doi.org/#1}{\path{#1}}}
\providecommand{\Pubmed}[1]{\href{pmid:#1}{\path{#1}}}
\providecommand{\bibinfo}[2]{#2}
\ifx\xfnm\relax \def\xfnm[#1]{\unskip,\space#1}\fi
\bibitem[{Abuin et~al.(2022)Abuin, Ferramosca, Toffanin, Magni and
  Gonzalez}]{ferramosca_insulin}
\bibinfo{author}{Abuin, P.}, \bibinfo{author}{Ferramosca, A.},
  \bibinfo{author}{Toffanin, C.}, \bibinfo{author}{Magni, L.},
  \bibinfo{author}{Gonzalez, A.H.}, \bibinfo{year}{2022}.
\newblock \bibinfo{title}{{Artificial pancreas under periodic MPC for
  trajectory tracking: handling circadian variability of insulin sensitivity}}.
\newblock \bibinfo{journal}{IFAC-PapersOnLine} \bibinfo{volume}{55},
  \bibinfo{pages}{196--201}.
\newblock \DOIprefix\doi{10.1016/j.ifacol.2022.09.023}.
\bibitem[{Alphonsus and Abdullah(2016)}]{alphonsus2016review}
\bibinfo{author}{Alphonsus, E.R.}, \bibinfo{author}{Abdullah, M.O.},
  \bibinfo{year}{2016}.
\newblock \bibinfo{title}{A review on the applications of programmable logic
  controllers ({PLC}s)}.
\newblock \bibinfo{journal}{Renewable and Sustainable Energy Reviews}
  \bibinfo{volume}{60}, \bibinfo{pages}{1185--1205}.
\newblock \DOIprefix\doi{10.1016/j.rser.2016.01.025}.
\bibitem[{Alvarado et~al.(2022)Alvarado, Krupa, Limon and
  Alamo}]{ALVARADO_ROBUST}
\bibinfo{author}{Alvarado, I.}, \bibinfo{author}{Krupa, P.},
  \bibinfo{author}{Limon, D.}, \bibinfo{author}{Alamo, T.},
  \bibinfo{year}{2022}.
\newblock \bibinfo{title}{Tractable robust {MPC} design based on nominal
  predictions}.
\newblock \bibinfo{journal}{Journal of Process Control} \bibinfo{volume}{111},
  \bibinfo{pages}{75--85}.
\newblock \DOIprefix\doi{10.1016/j.jprocont.2022.01.006}.
\bibitem[{Beck(2017)}]{beck_first-order_2017}
\bibinfo{author}{Beck, A.}, \bibinfo{year}{2017}.
\newblock \bibinfo{title}{First-Order Methods in Optimization}.
\newblock MOS-SIAM Series on Optimization, \bibinfo{publisher}{SIAM -- Society
  for Industrial and Applied Mathematics}.
\newblock \DOIprefix\doi{10.1137/1.9781611974997}.
\bibitem[{Boyd(2009)}]{Boyd_ConvexOptimization}
\bibinfo{author}{Boyd, S.}, \bibinfo{year}{2009}.
\newblock \bibinfo{title}{Convex Optimization}.
\newblock \bibinfo{edition}{7} ed., \bibinfo{publisher}{Cambridge University
  Press}.
\newblock \DOIprefix\doi{10.1017/CBO9780511804441}.
\bibitem[{Boyd et~al.(2011)Boyd, Parikh, Chu, Peleato and Eckstein}]{Boyd_ADMM}
\bibinfo{author}{Boyd, S.}, \bibinfo{author}{Parikh, N.}, \bibinfo{author}{Chu,
  E.}, \bibinfo{author}{Peleato, B.}, \bibinfo{author}{Eckstein, J.},
  \bibinfo{year}{2011}.
\newblock \bibinfo{title}{Distributed optimization and statistical learning via
  the alternating direction method of multipliers}.
\newblock \bibinfo{journal}{Foundations and Trends® in Machine Learning}
  \bibinfo{volume}{3}, \bibinfo{pages}{1--122}.
\newblock \DOIprefix\doi{10.1561/2200000016}.
\bibitem[{Camacho and Bordons(2007)}]{Camacho_S_2013}
\bibinfo{author}{Camacho, E.F.}, \bibinfo{author}{Bordons, C.},
  \bibinfo{year}{2007}.
\newblock \bibinfo{title}{{M}odel {P}redictive {C}ontrol}.
\newblock \bibinfo{edition}{2} ed., \bibinfo{publisher}{Springer London}.
\newblock \DOIprefix\doi{10.1007/978-0-85729-398-5}.
\bibitem[{Chen(1999)}]{Chen1995LinearST}
\bibinfo{author}{Chen, C.T.}, \bibinfo{year}{1999}.
\newblock \bibinfo{title}{Linear System Theory and Design}.
\newblock \bibinfo{edition}{3} ed., \bibinfo{publisher}{Oxford University
  Press}.
\bibitem[{Ferreau et~al.(2017)Ferreau, Almér, Verschueren, Diehl, Frick,
  Domahidi, Jerez, Stathopoulos and Jones}]{FERREAU201713194}
\bibinfo{author}{Ferreau, H.}, \bibinfo{author}{Almér, S.},
  \bibinfo{author}{Verschueren, R.}, \bibinfo{author}{Diehl, M.},
  \bibinfo{author}{Frick, D.}, \bibinfo{author}{Domahidi, A.},
  \bibinfo{author}{Jerez, J.}, \bibinfo{author}{Stathopoulos, G.},
  \bibinfo{author}{Jones, C.}, \bibinfo{year}{2017}.
\newblock \bibinfo{title}{Embedded optimization methods for industrial
  automatic control}.
\newblock \bibinfo{journal}{IFAC-PapersOnLine} \bibinfo{volume}{50},
  \bibinfo{pages}{13194--13209}.
\newblock \DOIprefix\doi{10.1016/j.ifacol.2017.08.1946}.
\bibitem[{Fletcher(2000)}]{Fletcher_PMO_2000}
\bibinfo{author}{Fletcher, R.}, \bibinfo{year}{2000}.
\newblock \bibinfo{title}{Practical Methods of Optimization}.
\newblock \bibinfo{publisher}{John Wiley \& Sons, Ltd}.
\newblock \DOIprefix\doi{10.1002/9781118723203}.
\bibitem[{Frison and Diehl(2020)}]{Frison_IFAC_2020}
\bibinfo{author}{Frison, G.}, \bibinfo{author}{Diehl, M.},
  \bibinfo{year}{2020}.
\newblock \bibinfo{title}{{HPIPM}: {A} high-performance quadratic programming
  framework for model predictive control}.
\newblock \bibinfo{journal}{IFAC-PapersOnLine} \bibinfo{volume}{53},
  \bibinfo{pages}{6563–6569}.
\newblock \DOIprefix\doi{10.1016/j.ifacol.2020.12.073}.
\bibitem[{Gracia et~al.(2024a)Gracia, Krupa, Limon and
  Alamo}]{Gracia_MPCT_solver_24}
\bibinfo{author}{Gracia, V.}, \bibinfo{author}{Krupa, P.},
  \bibinfo{author}{Limon, D.}, \bibinfo{author}{Alamo, T.},
  \bibinfo{year}{2024}a.
\newblock \bibinfo{title}{Efficient implementation of {MPC} for tracking using
  {ADMM} by decoupling its semi-banded structure}, in: \bibinfo{booktitle}{2024
  European Control Conference (ECC)}, pp. \bibinfo{pages}{2718--2723}.
\newblock \DOIprefix\doi{10.23919/ECC64448.2024.10591273}.
\bibitem[{Gracia et~al.(2024b)Gracia, Krupa, Limon and Alamo}]{Gracia_LCSS_24}
\bibinfo{author}{Gracia, V.}, \bibinfo{author}{Krupa, P.},
  \bibinfo{author}{Limon, D.}, \bibinfo{author}{Alamo, T.},
  \bibinfo{year}{2024}b.
\newblock \bibinfo{title}{Implementation of soft-constrained {MPC} for tracking
  using its semi-banded problem structure}.
\newblock \bibinfo{journal}{IEEE Control Systems Letters} \bibinfo{volume}{8},
  \bibinfo{pages}{1499--1504}.
\newblock \DOIprefix\doi{10.1109/LCSYS.2024.3407609}.
\bibitem[{He and Yuan(2012)}]{ADMM_convergence}
\bibinfo{author}{He, B.}, \bibinfo{author}{Yuan, X.}, \bibinfo{year}{2012}.
\newblock \bibinfo{title}{On the {$O(1/n)$} convergence rate of the
  {Douglas–Rachford} alternating direction method}.
\newblock \bibinfo{journal}{SIAM Journal on Numerical Analysis}
  \bibinfo{volume}{50}, \bibinfo{pages}{700--709}.
\newblock \DOIprefix\doi{10.1137/110836936}.
\bibitem[{Jerez et~al.(2014)Jerez, Goulart, Richter, Constantinides, Kerrigan
  and Morari}]{Goulart_embedded}
\bibinfo{author}{Jerez, J.L.}, \bibinfo{author}{Goulart, P.J.},
  \bibinfo{author}{Richter, S.}, \bibinfo{author}{Constantinides, G.A.},
  \bibinfo{author}{Kerrigan, E.C.}, \bibinfo{author}{Morari, M.},
  \bibinfo{year}{2014}.
\newblock \bibinfo{title}{Embedded online optimization for model predictive
  control at megahertz rates}.
\newblock \bibinfo{journal}{IEEE Transactions on Automatic Control}
  \bibinfo{volume}{59}, \bibinfo{pages}{3238--3251}.
\newblock \DOIprefix\doi{10.1109/TAC.2014.2351991}.
\bibitem[{Karg et~al.(2021)Karg, Alamo and Lucia}]{karg_prob_validation}
\bibinfo{author}{Karg, B.}, \bibinfo{author}{Alamo, T.},
  \bibinfo{author}{Lucia, S.}, \bibinfo{year}{2021}.
\newblock \bibinfo{title}{{Probabilistic performance validation of deep
  learning-based robust NMPC controllers}}.
\newblock \bibinfo{journal}{International Journal of Robust and Nonlinear
  Control} \bibinfo{volume}{31}, \bibinfo{pages}{8855--8876}.
\newblock \DOIprefix\doi{10.1002/rnc.5696}.
\bibitem[{Kerrigan and Maciejowski(2000)}]{Kerrigan2000SoftCA}
\bibinfo{author}{Kerrigan, E.C.}, \bibinfo{author}{Maciejowski, J.M.},
  \bibinfo{year}{2000}.
\newblock \bibinfo{title}{Soft constraints and exact penalty functions in model
  predictive control}, in: \bibinfo{booktitle}{Control 2000 Conference,
  Cambridge}, pp. \bibinfo{pages}{2319--2327}.
\newblock \bibinfo{note}{\url{http://hdl.handle.net/10044/1/10241}}.
\bibitem[{Krupa et~al.(2021a)Krupa, Camara, Alvarado, Limon and
  Alamo}]{Krupa_ECC_21}
\bibinfo{author}{Krupa, P.}, \bibinfo{author}{Camara, J.},
  \bibinfo{author}{Alvarado, I.}, \bibinfo{author}{Limon, D.},
  \bibinfo{author}{Alamo, T.}, \bibinfo{year}{2021}a.
\newblock \bibinfo{title}{{Real-time implementation of MPC for tracking in
  embedded systems: Application to a two-wheeled inverted pendulum}}, in:
  \bibinfo{booktitle}{2021 European Control Conference (ECC)}, pp.
  \bibinfo{pages}{669--674}.
\newblock \DOIprefix\doi{10.23919/ECC54610.2021.9654899}.
\bibitem[{Krupa et~al.(2020)Krupa, Gracia, Limon and Alamo}]{Spcies}
\bibinfo{author}{Krupa, P.}, \bibinfo{author}{Gracia, V.},
  \bibinfo{author}{Limon, D.}, \bibinfo{author}{Alamo, T.},
  \bibinfo{year}{2020}.
\newblock \bibinfo{title}{{SPCIES: Suite of predictive controllers for
  industrial embedded systems}}.
\newblock \bibinfo{howpublished}{\url{https://github.com/GepocUS/Spcies}}.
\bibitem[{Krupa et~al.(2024)Krupa, K{\"o}hler, Ferramosca, Alvarado, Zeilinger,
  Alamo and Limon}]{krupa_tutorial_2024}
\bibinfo{author}{Krupa, P.}, \bibinfo{author}{K{\"o}hler, J.},
  \bibinfo{author}{Ferramosca, A.}, \bibinfo{author}{Alvarado, I.},
  \bibinfo{author}{Zeilinger, M.N.}, \bibinfo{author}{Alamo, T.},
  \bibinfo{author}{Limon, D.}, \bibinfo{year}{2024}.
\newblock \bibinfo{title}{Model predictive control for tracking using
  artificial references: Fundamentals, recent results and practical
  implementation}, in: \bibinfo{booktitle}{2024 IEEE 63rd Conference on
  Decision and Control (CDC)}, pp. \bibinfo{pages}{2977--2991}.
\newblock \DOIprefix\doi{10.1109/CDC56724.2024.10886854}.
\bibitem[{Krupa et~al.(2021b)Krupa, Limon and Alamo}]{Krupa_TCST_21}
\bibinfo{author}{Krupa, P.}, \bibinfo{author}{Limon, D.},
  \bibinfo{author}{Alamo, T.}, \bibinfo{year}{2021}b.
\newblock \bibinfo{title}{{Implementation of model predictive control in
  programmable logic controllers}}.
\newblock \bibinfo{journal}{IEEE Transactions on Control Systems Technology}
  \bibinfo{volume}{29}, \bibinfo{pages}{1117--1130}.
\newblock \DOIprefix\doi{10.1109/TCST.2020.2992959}.
\bibitem[{Langson et~al.(2004)Langson, Chryssochoos, Raković and
  Mayne}]{LANGSON}
\bibinfo{author}{Langson, W.}, \bibinfo{author}{Chryssochoos, I.},
  \bibinfo{author}{Raković, S.}, \bibinfo{author}{Mayne, D.},
  \bibinfo{year}{2004}.
\newblock \bibinfo{title}{Robust model predictive control using tubes}.
\newblock \bibinfo{journal}{Automatica} \bibinfo{volume}{40},
  \bibinfo{pages}{125--133}.
\newblock \DOIprefix\doi{10.1016/j.automatica.2003.08.009}.
\bibitem[{Limon et~al.(2008)Limon, Alvarado, Alamo and Camacho}]{LIMON20082382}
\bibinfo{author}{Limon, D.}, \bibinfo{author}{Alvarado, I.},
  \bibinfo{author}{Alamo, T.}, \bibinfo{author}{Camacho, E.},
  \bibinfo{year}{2008}.
\newblock \bibinfo{title}{{MPC} for tracking piecewise constant references for
  constrained linear systems}.
\newblock \bibinfo{journal}{Automatica} \bibinfo{volume}{44},
  \bibinfo{pages}{2382--2387}.
\newblock \DOIprefix\doi{10.1016/j.automatica.2008.01.023}.
\bibitem[{Lowenstein et~al.(2024)Lowenstein, Bernardini and
  Patrinos}]{Lowenstein_LCSS_24}
\bibinfo{author}{Lowenstein, K.F.}, \bibinfo{author}{Bernardini, D.},
  \bibinfo{author}{Patrinos, P.}, \bibinfo{year}{2024}.
\newblock \bibinfo{title}{{QPALM-OCP: A Newton-type proximal augmented
  lagrangian solver tailored for quadratic programs arising in model predictive
  control}}.
\newblock \bibinfo{journal}{IEEE Control Systems Letters} \bibinfo{volume}{8},
  \bibinfo{pages}{1349–1354}.
\newblock \DOIprefix\doi{10.1109/LCSYS.2024.3410638}.
\bibitem[{Maeder et~al.(2009)Maeder, Borrelli and Morari}]{MAEDER}
\bibinfo{author}{Maeder, U.}, \bibinfo{author}{Borrelli, F.},
  \bibinfo{author}{Morari, M.}, \bibinfo{year}{2009}.
\newblock \bibinfo{title}{Linear offset-free model predictive control}.
\newblock \bibinfo{journal}{Automatica} \bibinfo{volume}{45},
  \bibinfo{pages}{2214--2222}.
\newblock \DOIprefix\doi{10.1016/j.automatica.2009.06.005}.
\bibitem[{Mahnke et~al.(2009)Mahnke, Leitner and Damm}]{OPC_UA}
\bibinfo{author}{Mahnke, W.}, \bibinfo{author}{Leitner, S.H.},
  \bibinfo{author}{Damm, M.}, \bibinfo{year}{2009}.
\newblock \bibinfo{title}{{OPC} Unified Architecture}.
\newblock \bibinfo{publisher}{Springer-Verlag Berlin Heidelberg}.
\newblock \DOIprefix\doi{10.1007/978-3-540-68899-0}.
\bibitem[{Mammarella et~al.(2020)Mammarella, Alamo, Lucia and
  Dabbene}]{mammarella}
\bibinfo{author}{Mammarella, M.}, \bibinfo{author}{Alamo, T.},
  \bibinfo{author}{Lucia, S.}, \bibinfo{author}{Dabbene, F.},
  \bibinfo{year}{2020}.
\newblock \bibinfo{title}{A probabilistic validation approach for penalty
  function design in stochastic model predictive control}.
\newblock \bibinfo{journal}{IFAC-PapersOnLine} \bibinfo{volume}{53},
  \bibinfo{pages}{11271--11276}.
\newblock \DOIprefix\doi{10.1016/j.ifacol.2020.12.362}.
\bibitem[{Moscato et~al.(2024)Moscato, Sanalitro, Stella and
  Bucolo}]{Moscato_CEP_2024}
\bibinfo{author}{Moscato, S.}, \bibinfo{author}{Sanalitro, D.},
  \bibinfo{author}{Stella, G.}, \bibinfo{author}{Bucolo, M.},
  \bibinfo{year}{2024}.
\newblock \bibinfo{title}{Model predictive control framework for slug flow
  microfluidics processes}.
\newblock \bibinfo{journal}{Control Engineering Practice}
  \bibinfo{volume}{148}, \bibinfo{pages}{105944}.
\newblock \DOIprefix\doi{10.1016/j.conengprac.2024.105944}.
\bibitem[{Nubert et~al.(2020)Nubert, Köhler, Berenz, Allgöwer and
  Trimpe}]{Nubert_RAL_2020}
\bibinfo{author}{Nubert, J.}, \bibinfo{author}{Köhler, J.},
  \bibinfo{author}{Berenz, V.}, \bibinfo{author}{Allgöwer, F.},
  \bibinfo{author}{Trimpe, S.}, \bibinfo{year}{2020}.
\newblock \bibinfo{title}{Safe and fast tracking on a robot manipulator:
  {R}obust {MPC} and neural network control}.
\newblock \bibinfo{journal}{IEEE Robotics and Automation Letters}
  \bibinfo{volume}{5}, \bibinfo{pages}{3050–3057}.
\newblock \DOIprefix\doi{10.1109/LRA.2020.2975727}.
\bibitem[{O'Donoghue(2021)}]{ODonoghue_SCS_21}
\bibinfo{author}{O'Donoghue, B.}, \bibinfo{year}{2021}.
\newblock \bibinfo{title}{Operator splitting for a homogeneous embedding of the
  linear complementarity problem}.
\newblock \bibinfo{journal}{{SIAM} Journal on Optimization}
  \bibinfo{volume}{31}, \bibinfo{pages}{1999--2023}.
\newblock \DOIprefix\doi{10.1137/20M1366307}.
\bibitem[{Rawlings et~al.(2017)Rawlings, Mayne and Diehl}]{rawlings_model_2017}
\bibinfo{author}{Rawlings, J.B.}, \bibinfo{author}{Mayne, D.Q.},
  \bibinfo{author}{Diehl, M.}, \bibinfo{year}{2017}.
\newblock \bibinfo{title}{{Model Predictive Control: Theory, Computation, and
  Design}}.
\newblock \bibinfo{edition}{2} ed., \bibinfo{publisher}{Nob Hill Publishing}.
\newblock \bibinfo{note}{\url{https://sites.engineering.ucsb.edu/~jbraw/mpc}}.
\bibitem[{Santos et~al.(2019)Santos, Bonzanini, Heirung and
  Mesbah}]{constraint_tightening}
\bibinfo{author}{Santos, T.L.}, \bibinfo{author}{Bonzanini, A.D.},
  \bibinfo{author}{Heirung, T.A.N.}, \bibinfo{author}{Mesbah, A.},
  \bibinfo{year}{2019}.
\newblock \bibinfo{title}{A constraint-tightening approach to nonlinear model
  predictive control with chance constraints for stochastic systems}, in:
  \bibinfo{booktitle}{2019 American Control Conference (ACC)}, pp.
  \bibinfo{pages}{1641--1647}.
\newblock \DOIprefix\doi{10.23919/ACC.2019.8814623}.
\bibitem[{Stellato et~al.(2020)Stellato, Banjac, Goulart, Bemporad and
  Boyd}]{stellato_osqp_2020}
\bibinfo{author}{Stellato, B.}, \bibinfo{author}{Banjac, G.},
  \bibinfo{author}{Goulart, P.}, \bibinfo{author}{Bemporad, A.},
  \bibinfo{author}{Boyd, S.}, \bibinfo{year}{2020}.
\newblock \bibinfo{title}{{OSQP}: An operator splitting solver for quadratic
  programs}.
\newblock \bibinfo{journal}{Mathematical Programming Computation}
  \bibinfo{volume}{12}, \bibinfo{pages}{637–672}.
\newblock \DOIprefix\doi{10.1007/s12532-020-00179-2}.
\bibitem[{Subramanian et~al.(2015)Subramanian, Lucia and
  Engel}]{CSTR_SUBRAMANIAN}
\bibinfo{author}{Subramanian, S.}, \bibinfo{author}{Lucia, S.},
  \bibinfo{author}{Engel, S.}, \bibinfo{year}{2015}.
\newblock \bibinfo{title}{Adaptive multi-stage output feedback {NMPC} using the
  extended {Kalman} filter for time varying uncertainties applied to a cstr}.
\newblock \bibinfo{journal}{IFAC-PapersOnLine} \bibinfo{volume}{48},
  \bibinfo{pages}{242--247}.
\newblock \DOIprefix\doi{10.1016/j.ifacol.2015.11.290}.
\bibitem[{Tylavsky and Sohie(1986)}]{woodbury}
\bibinfo{author}{Tylavsky, D.}, \bibinfo{author}{Sohie, G.},
  \bibinfo{year}{1986}.
\newblock \bibinfo{title}{Generalization of the matrix inversion lemma}.
\newblock \bibinfo{journal}{Proceedings of the IEEE} \bibinfo{volume}{74},
  \bibinfo{pages}{1050--1052}.
\newblock \DOIprefix\doi{10.1109/PROC.1986.13587}.
\bibitem[{Valencia-Palomo and Rossiter(2011)}]{Valencia_CEP_2011}
\bibinfo{author}{Valencia-Palomo, G.}, \bibinfo{author}{Rossiter, J.},
  \bibinfo{year}{2011}.
\newblock \bibinfo{title}{Efficient suboptimal parametric solutions to
  predictive control for {PLC} applications}.
\newblock \bibinfo{journal}{Control Engineering Practice} \bibinfo{volume}{19},
  \bibinfo{pages}{732–743}.
\newblock \DOIprefix\doi{10.1016/j.conengprac.2011.04.001}.
\bibitem[{Wabersich et~al.(2022)Wabersich, Krishnadas and
  Zeilinger}]{Wabersich2022980}
\bibinfo{author}{Wabersich, K.P.}, \bibinfo{author}{Krishnadas, R.},
  \bibinfo{author}{Zeilinger, M.N.}, \bibinfo{year}{2022}.
\newblock \bibinfo{title}{A soft constrained {MPC} formulation enabling
  learning from trajectories with constraint violations}.
\newblock \bibinfo{journal}{IEEE Control Systems Letters} \bibinfo{volume}{6},
  \bibinfo{pages}{980 – 985}.
\newblock \DOIprefix\doi{10.1109/LCSYS.2021.3087968}.

\end{thebibliography}


\end{document}